\documentclass[runningheads]{llncs}

\usepackage{hyperref}
\usepackage{wrapfig}
\usepackage{array}
\usepackage{stmaryrd}
\usepackage{amsmath,amssymb}
\usepackage{proof}
\usepackage{xcolor}
\usepackage{macros}

\begin{document}
\title{Reachability is Decidable for ATM-Typable Finitary PCF with Effect Handlers}
\titlerunning{Finitary PCF with Effect Handlers and ATM has Decidable Reachability}

\renewcommand*{\thefootnote}{\fnsymbol{footnote}}

\author{Ryunosuke Endo\inst{1} \and Tachio Terauchi\inst{2}}
\institute{Waseda University, Tokyo, Japan\footnote{Author's current affiliation: Mizuho Bank, Ltd.} \\
  \email{minerva@ruri.waseda.jp} \and
  Waseda University, Tokyo, Japan\\
  \email{terauchi@waseda.jp}}

\maketitle

\renewcommand*{\thefootnote}{\arabic{footnote}}
\setcounter{footnote}{0}

\begin{abstract}

It is well known that the reachability problem for simply-typed lambda calculus with recursive definitions and finite base-type values (finitary PCF) is decidable.  A recent paper by Dal Lago and Ghyselen has shown that the same problem becomes undecidable when the language is extended with algebraic effect and handlers (effect handlers).
We show that, perhaps surprisingly, the problem becomes decidable even with effect handlers when the type system is extended with answer type modification (ATM).  A natural intuition may find the result contradictory, because one would expect allowing ATM makes more programs typable. Indeed, this intuition is correct in that there are programs that are typable with ATM but not without it, as we shall show in the paper.  However, a corollary of our decidability result is that the converse is true as well: there are programs that are typable without ATM but becomes untypable with ATM, and we will show concrete examples of such programs in the paper.  Our decidability result is proven by a novel continuation passing style (CPS) transformation that transforms an ATM-typable finitary PCF program with effect handlers to a finitary PCF program without effect handlers.
Additionally, as another application of our CPS transformation, we show that every recursive-function-free ATM-typable finitary PCF program with effect handlers terminates, while there are (necessarily ATM-untypable) recursive-function-free finitary PCF programs with effect handlers that may diverge.  Finally, we disprove a claim made in a recent work that proved a similar but strictly weaker decidability result.
We foresee our decidability result to lay a foundation for developing verification methods for programs with effect handlers, just as the decidability result for reachability of finitary PCF has done such for programs without effect handlers.
\end{abstract}


\section{Introduction}

A popular approach to the verification of infinite-state programs is to abstract the programs so that their base-type values are over finite domains, as seen in, for example, predicate abstraction with CEGAR~\cite{DBLP:conf/popl/BallR02,DBLP:conf/cav/ClarkeGJLV00,DBLP:conf/pldi/BallMMR01,DBLP:conf/popl/Terauchi10,DBLP:conf/popl/OngR11,DBLP:conf/pldi/KobayashiSU11,DBLP:conf/popl/UnnoTK13,DBLP:conf/popl/RamsayNO14,DBLP:conf/lics/PodelskiR04,DBLP:conf/pldi/CookPR06,DBLP:conf/popl/CookGPRV07,DBLP:conf/esop/KuwaharaTU014,DBLP:conf/popl/MuraseT0SU16}.  Importantly, the reachability problem, which asks whether there exists an execution of the program reaching a certain program state (typically a designated ``error'' state), is known to be decidable for such programs when they are simply-typable, even when the programs contain higher-order and recursive functions.  That is, reachability for finitary PCF is decidable~\cite{DBLP:conf/lics/Ong06,DBLP:conf/popl/Kobayashi09,DBLP:conf/fossacs/Tsukada014}.\footnote{The result can be seen as an extension of the decidability result for reachability of pushdown systems~\cite{DBLP:conf/cav/EsparzaHRS00}, which correspond to first-order recursive programs, to higher-order recursive programs.  Also, the result should not be confused with the result that observational equivalence for finitary PCF is undecidable~\cite{DBLP:journals/tcs/Loader01}.}

Meanwhile, algebraic effects and handlers (\emph{effect handlers} henceforth) are a programming language feature for expressing computational effects such as mutable state, exception, and non-determinism~\cite{DBLP:journals/entcs/Pretnar15}.  They have a theoretical origin, stemming from the research on denotational semantics~\cite{DBLP:conf/fossacs/PlotkinP01,DBLP:conf/fossacs/PlotkinP02,DBLP:journals/acs/PlotkinP03,DBLP:conf/lics/PlotkinP08,DBLP:conf/esop/PlotkinP09,DBLP:journals/corr/abs-1807-05923}, but are also practically popular and have been incorporated into many popular programming languages such as C, C++, Java, OCaml, and Haskell~\cite{DBLP:conf/aplas/Leijen17,DBLP:journals/pacmpl/BrachthauserSO18,DBLP:conf/haskell/XieL20,DBLP:conf/pldi/Sivaramakrishnan21,DBLP:journals/pacmpl/GhicaLBP22,DBLP:journals/pacmpl/Alvarez-Picallo24}.  Unfortunately, a recent paper~\cite{DBLP:journals/pacmpl/LagoG24} has shown that reachability, which is decidable for finitary PCF as mentioned above, becomes undecidable when the language is extended with effect handlers.

In this paper, we show that this undecidability stems from the way the standard type systems for effect handlers are designed.  More concretely, we show that, perhaps surprisingly, extending the standard type system to allow answer type modification (ATM)~\cite{DBLP:conf/lfp/DanvyF90,DBLP:journals/lisp/Asai09,DBLP:journals/pacmpl/KawamataUST24} recovers decidability.  A natural intuition may find the result contradictory because one would expect allowing ATM makes more programs typable.  Indeed, this intuition is correct in that there are programs that are typable with ATM but not without it, as we shall show in the paper (cf.~Example~\ref{ex:stuntypableatmtypable}).  However, a corollary of our decidability result is that the converse is true as well: there are programs that are typable without ATM but becomes untypable with ATM, and we will show concrete examples of such programs in the paper (cf.~Sections~\ref{subsec:atmuntypableex} and \ref{subsec:termination}).

Our decidability result is proven by a novel continuation passing style (CPS) transformation that transforms an ATM-typable finitary PCF program with effect handlers to a finitary PCF program without effect handlers.  Then, our decidability result follows from the fact that the target of the CPS transformation, finitary PCF, has decidable reachability as mentioned in the first paragraph of this section.

Additionally, as another application of our CPS transformation, we show that every recursive-function-free ATM-typable finitary PCF program with effect handlers terminates, while there are (necessarily ATM-untypable) recursive-function-free finitary PCF programs with effect handlers that may diverge.  Finally, we disprove in a claim made in a recent work that proved a similar but strictly weaker decidability result~\cite{DBLP:journals/pacmpl/Sekiyama024}.  In summary, the main contributions of the paper are as follows.
\begin{itemize}
\item We show that reachability for ATM-typable finitary PCF with effect handlers is decidable.  A novel CPS transformation is introduced to prove the result.
\item A corollary of our decidability result is that there are finitary PCF programs typable without ATM but untypable with it, and we show concrete examples of such programs.
\item As another application of the CPS transformation, we show that recursive-function-free ATM-typable finitary PCF programs with effect handlers always terminate, while there are (ATM-untypable) recursive-function-free finitary PCF programs with effect handlers that may diverge.
\item We disprove a claim made in a recent paper~\cite{DBLP:journals/pacmpl/Sekiyama024} regarding what the paper calls the number of ``active effect handlers''.
\end{itemize}
The rest of the paper is organized as follows.  Section~\ref{sec:prelim} defines preliminary notions.  Section~\ref{sec:mainresults} contains the main results mentioned above.  Section~\ref{sec:related} discusses related work.  Section~\ref{sec:conc} concludes the paper.  For space, some materials are deferred to \full{the appendix}{the extended report~\cite{fullversion}}.

\section{Preliminaries}
\label{sec:prelim}
\begin{figure}[t]
\[
\begin{array}{rcl}
v & \mathrm{::=} & x \mid () \mid \tttrue \mid \ttfalse \mid \lambda x.c \mid \ttrec{x}{v} \\
c & \mathrm{::=} & \ttreturn\;v \mid \itop\;v \mid v_1\:v_2 \mid \ttif{v}{c_1}{c_2} \mid \ttlet{x}{c_1}{c_2} \mid \ttwithhandle{h}{c} \\
h & \mathrm{::=} & \{ \ttreturn\:x = c, \overline{\itop_i\:x_i\:k_i = c_i} \}
\end{array}
\]
\caption{The syntax of $\plname$.}
\label{fig:syntax}
\end{figure}

Figure~\ref{fig:syntax} shows the syntax of untyped finitary PCF with effect handlers $\plname$.  As in many presentations of effect handlers~\cite{DBLP:journals/entcs/Pretnar15,DBLP:journals/pacmpl/KawamataUST24}, we adopt the approach of call-by-push-value $\lambda$ calculus~\cite{DBLP:books/sp/Levy2004} to separate the syntax of expressions into \emph{values}, ranged over by meta variables $v$, and \emph{computations}, ranged over by $c$.  A {\em handler}, ranged over by $h$, consists of a single \emph{return clause} of the form $\ttreturn\;x = c$ and finitely many \emph{operation clauses} of the form $\itop_i\;x_i\;k_i = c_i$.  The parameter $k_i$ in an operation clause is called the \emph{continuation parameter}.
Recursive definitions are given by the syntax $\ttrec{x}{v}$ which recursively binds $x$ in $v$.
\texcomment{
We also often use the notation
\[
\mathtt{letrec}\;x_1 = v_1\;\mathtt{and}\;x_2 = v_2\;\mathtt{and}\;\cdots\;x_n = v_n\;\mathtt{in}\;v
\]
for mutually recursive definition, which is a derived form defined inductively as:
\[
\begin{array}{l}
  \mathtt{let}\:x_1 = \ttreturn\:(\ttrec{x_1}{\mathtt{letrec}\:x_2 = v_2\:\mathtt{and}\:\cdots\:\mathtt{and}\:x_n = v_n\:\mathtt{in}\:v_1}\:\mathtt{in}\\
  \cdots \\
  \mathtt{let}\:x_n = \ttreturn\:(\ttrec{x_n}{\mathtt{letrec}\:x_1 = v_1\:\mathtt{and}\:\cdots\:\mathtt{and}\:x_{n-1} = v_{n-1}\:\mathtt{in}\:v_n}\:\mathtt{in}\\
  \ttreturn\:v
\end{array}
\]
}

Note that functions, arguments, and conditional expressions are restricted to values.  This does not reduce expressivity because, for example, a function application $c\;v$ can be expressed as $\ttlet{x}{c}{x\;v}$ using a fresh variable $x$.  For convenience, we often assume this convention and allow computations to appear at positions where values are expected.  For example, we may write $x\;y\;z$ for $\ttlet{w}{x\;y}{w\;z}$, adopting the usual convention that function application is left associative.
Conversely, when a value $v$ appears at a position where a computation in expected, we read it as $\ttreturn\;v$.  For example, we may write $\lambda x.\lambda y.x$ for $\lambda x.\ttreturn\;\lambda y.\ttreturn\;x$.
We also write $c_1;c_2$ for $\ttlet{x}{c_1}{c_2}$ where $x$ is not free in $c_2$.
As usual, a \emph{program} is a closed expression.

For concreteness, $\plname$ uses Booleans and unit as base-type values, but our results can be easily adopted to other finite base-type domains such as variant types and integers modulo a constant.  Note that $\plname$ is untyped.  Later in the paper, we present two type systems for it, an ordinary simple type system $\vdashst$ and a type system with answer-type modification $\vdashatm$, and investigate how they affect the decidability of reachability.

\begin{figure}[t]
\[
\begin{array}{c}
  \infer[(\textsc{E-Let})]{\ttlet{x}{c_1}{c} \rightarrow \ttlet{x}{c_2}{c}}{c_1 \rightarrow c_2} \ \
  \infer[(\textsc{E-Ret})]{\ttlet{x}{\ttreturn\:v}{c} \rightarrow c[v/x]}{} \\\\
  \infer[(\textsc{E-LamApp})]{(\lambda x.c)\:v \rightarrow c[v/x]}{} \ \
  \infer[(\textsc{E-RecApp})]{(\ttrec{x}{v})\:v' \rightarrow v[(\ttrec{x}{v})/x]\:v'}{} \\\\
    \infer[(\textsc{E-IfTrue})]{\ttif{\tttrue}{c_1}{c_2} \rightarrow c_1}{} \ \
    \infer[(\textsc{E-IfFalse})]{\ttif{\ttfalse}{c_1}{c_2} \rightarrow c_2}{} \\\\
    \infer[(\textsc{E-Han})]{\ttwithhandle{h}{c} \rightarrow \ttwithhandle{h}{c'}}{c \rightarrow c'} \\\\
    \infer[(\textsc{E-HRet})]{\ttwithhandle{h}{\ttreturn\:v} \rightarrow c[v/x]}{\ttreturn\:x = c \in h}\\\\
    \infer[(\textsc{E-Op})]{\ttwithhandle{h}{E[\itop\:v]} \rightarrow c[v/x,\lambda y.\ttwithhandle{h}{E[\ttreturn\:y]}/k]}{\itop\:x\:k = c \in h}
\end{array}
\]
\caption{The semantics of $\plname$.}
\label{fig:semantics}
\end{figure}

Figure~\ref{fig:semantics} shows the operational semantics of $\plname$.  Here, the \emph{evaluation context} $E$ is defined by: $E\:\textrm{::=}\:[\:] \mid \ttlet{x}{E}{c}$.  The semantics is standard for a language with effect handlers.  A key rule is $(\textsc{E-Op})$ which invokes an operation.  An operation invocation is quite different from an ordinary function call, and replaces the current context up to and including the nearest with-handle block with the body of the operation clause $c$.  The actual argument $v$ gets bound to the formal parameter $x$ and the continuation parameter $k$ gets the captured \emph{delimited continuation} $\lambda y.\ttwithhandle{h}{E[\ttreturn\;y]}$. The behavior is similar to that of the \emph{shift} operation from delimited control~\cite{DBLP:conf/lfp/DanvyF90,DBLP:journals/lisp/Asai09}.  Note that the evaluation context $E$ in the delimited continuation is wrapped in the with-handle block, following the standard \emph{deep-handler} semantics~\cite{DBLP:conf/icfp/KammarLO13}.
%
%
Another important rule is $(\textsc{E-HRet})$ which processes a return clause invocation.  Note that unlike the ordinary return of $(\textsc{E-Ret})$, a return at a tail position of a with-handle block invokes the return clause so that the returned result is (the evaluation result of) $c[v/x]$ rather than $v$.
%
%
As standard, we write $\rightarrow^*$ to denote the reflexive transitive closure of $\rightarrow$.  The \emph{reachability problem} is defined as follows.
\begin{definition}[Reachability]
  \normalfont
The \emph{reachability problem} is to decide, given a program $c$, if $c \rightarrow^* \ttreturn\:\tttrue$.
\end{definition}
We note that the choice of $\tttrue$ is arbitrary. We could have alternatively chosen any other base-type value as the final value that we would like to decide if the program returns or not.

\begin{example}
\label{ex:loop}
Let $c_{\it ex1}$ be the following program.\footnote{Recall the syntactic sugar such as the notation $c_1;c_2$ remarked earlier.}
\[
\begin{array}{l}
  (\mathtt{with}\;h_{\it state}\;\mathtt{handle}\;(\mathtt{rec}\:\mathit{loop} = \;\lambda \_.\\
  \hspace{4cm} \mathtt{let}\;n = \mathit{get}\;()\;\mathtt{in}\\
  \hspace{4cm} \mathtt{if}\;v_{\it fst}\;n\;\mathtt{then}\;\mathtt{if}\;v_{\it snd}\;n\;\mathtt{then}\;\;c_{\it incr};\mathit{loop}\;()\;\mathtt{else}\;()\\
  \hspace{4cm} \mathtt{else}\;c_{\it incr};\mathit{loop}\;())\;(v_{\it tup}\;\ttfalse\;\ttfalse)); \ttreturn\;\tttrue
\end{array}
\]
where
\[
\begin{array}{rcl}
  v_{\it tup} & \defeq & \lambda x.\lambda y.\lambda f.f\;x\;y \ \ \
  v_{\it fst} \defeq \lambda p.p\;\lambda x.\lambda y.x \ \ \
  v_{\it snd} \defeq \lambda p.p\;\lambda x.\lambda y.y\\
  c_{\it inc} & \defeq & \mathtt{let}\;n =  \mathit{get}\;()\;\mathtt{in}\\
  & & \mathtt{if}\;v_{\it fst}\;n\;\mathtt{then}\;\mathtt{if}\;v_{\it snd}\;n\;\mathtt{then}\;()\;\mathtt{else}\;\mathit{set}\;(v_{\it tup}\;\ttfalse\;\tttrue) \\
  & & \mathtt{else}\;\mathtt{if}\;v_{\it snd}\;n\;\mathtt{then}\;\mathit{set}\;(v_{\it tup}\;\tttrue\;\tttrue)\;\mathtt{else}\;\mathit{set}\;(v_{\it tup}\;\tttrue\;\ttfalse) \\
  h_{\it state} & \defeq & \{ \ttreturn\;x = \lambda s.x, \mathit{set}\;x\;k = \lambda s.k\;()\;x, \mathit{get}\;x\;k = \lambda s.k\;s\;s \} 
\end{array}
\]
The handler $h_{\it state}$ adopts the standard state-passing approach for expressing mutable states with effect handlers~\cite{DBLP:journals/entcs/Pretnar15}.  Namely, the operation $\mathit{set}$ updates the current state with the given argument and the operation $\mathit{get}$ returns the current state.  In this program, a state is a pair of Booleans encoding a 2-bit non-negative integer.  We use the standard $\lambda$ calculus encoding of pairs: $v_{\it tup}$, $v_{\it fst}$, and $v_{\it snd}$ respectively creates a pair, projects the first element, and the projects the second element.  The computation $c_{\it inc}$ is effectful and uses $\mathit{get}$ and $\mathit{set}$ to increment the current state by one.  The initial state is set to be zero (i.e., the pair $(\bot,\bot)$), and the recursive function defined by $\ttrec{\mathit{loop}}{\dots}$ repeatedly increments the state until the value becomes one (i.e., $(\top,\bot)$). Because the initial state is zero, one is eventually reached and the program $c_{\it ex1}$ terminates returning $\top$.  The same program would also terminate and return $\top$ if the state was initialized to be one, but it would diverge if the state was initialized to be two (i.e., $(\bot,\top)$) or three (i.e., $(\top,\top)$).  Therefore, the answer to the reachability problem for this program would be yes in the first two cases and be no in the latter two cases.
\end{example}

\begin{example}
\label{ex:nondet}
  
Next, let $c_{\it ex2}$ be the following program, also adopted from \cite{DBLP:journals/entcs/Pretnar15}.
\[
\begin{array}{l}
  
  \mathtt{let}\;r =\;(\mathtt{with}\;h_{\it nd}\;\mathtt{handle}\;\mathtt{let}\;x=\;\ttif{\mathit{dec}\;()}{v_0}{v_1}\;\mathtt{in}\\
  \hspace{4.25cm} \mathtt{let}\;y =\;\ttif{\mathit{dec}\;()}{v_2}{v_3}\;\mathtt{in} \; v_{\it xor}\;x\;y)\;\mathtt{in}\\
  \ttif{r}{(\ttrec{f}{\lambda \_.f\;()})\;()}{\ttreturn\;\tttrue}
\end{array}
\]
where
\[
\begin{array}{rcl}
  v_{\it xor} & \defeq & \lambda x.\lambda y.\ttif{x}{(\ttif{y}{\ttfalse}{\tttrue})}{(\ttif{y}{\tttrue}{\ttfalse})} \\
  v_{\it or} & \defeq & \lambda x.\lambda y.\ttif{x}{\tttrue}{\ttif{y}{\tttrue}{\ttfalse}} \\
  h_{\it nd} & \defeq & \{ \ttreturn\;x=x, \mathit{dec}\;x\;k = v_{\it or}\;(k\;\tttrue)\;(k\;\ttfalse) \}
  \end{array}
\]
Here, $v_0, v_1, v_2, v_3 \in \{ \ttfalse,\tttrue \}$.  The handler $h_{\it nd}$ implements non-deterministic choice by the $\mathit{dec}$ operation whose clause executes the given delimited continuation $k$, which is expected to take and return Booleans, with both $\tttrue$ an $\ttfalse$, and returns the disjunction of the two results.  Therefore, the with-handle block will, for each of the four possibilities where $x$ is bound to $v_0$ or $v_1$ and $y$ is bound to $v_2$ or $v_3$, computes the exclusive-or of $x$ and $y$, and takes the disjunction of the four results.  Therefore, the with-handle block returns $\ttfalse$ if the Booleans $v_0, v_1, v_2, v_3$ are all equal and otherwise returns $\tttrue$, and the program $c_{\it ex2}$ returns $\tttrue$ in the former cases and diverges in the later cases (due to the infinite loop $(\ttrec{f}{\lambda \_.f\;()})\;()$).  Thus, the answer to the reachability problem would be true in the former two cases (i.e., the four Booleans are all $\tttrue$ or all $\ttfalse$) and be no in the latter 14 cases.
\end{example}

\begin{wrapfigure}[5]{r}[0pt]{0.33\textwidth}
    \vspace{-30pt}
\[
\begin{array}{rcl}
  b & \textrm{::=} & \unitty \mid \boolty \\
  \sigma & \textrm{::=} & b \mid \sigma_1 \rightarrow \sigma_2
\end{array}
\]
  \vspace{-20pt}
\caption{The simple types.}
\label{fig:stype}
\end{wrapfigure}
Figure~\ref{fig:stype} defines the simple types. The base types are denoted by $b$.  As usual, the function type constructor $\rightarrow$ associates to the right.
Figure~\ref{fig:stypingrules} shows the typing rules of the simple type system $\vdashst$.    The type system is parameterized by a \emph{signature} $\Sigma$ that assigns types to the operation names.
The rules $(\textsc{St-Unit})$ to $(\textsc{St-Let})$ are standard.  The last four rules concern effect handlers, and they are also standard, matching those studied in prior work~\cite{DBLP:journals/entcs/Pretnar15,DBLP:journals/pacmpl/LagoG24,DBLP:journals/pacmpl/Kobayashi25}.\footnote{Technically, \cite{DBLP:journals/pacmpl/LagoG24,DBLP:journals/pacmpl/Kobayashi25} use weaker type systems that restrict (non-continuation) parameters of operations to base types.  Our main result shows that ATM-typability makes reachability decidable even when no such restriction is imposed (cf.~Sections~\ref{subsec:errorinsekiyamaunno} and \ref{sec:related} for further discussion).}
Importantly, $\textsc{(St-Hdlr)}$ checks if a handler is well-typed by checking the well-typedness of the return clause as well as that of each operation clause.  Note that the type of the return clause body, the types of the operation clause bodies, and the return types of the continuation parameters, are the same type $\sigma'$.  This type $\sigma'$ is called the \emph{answer type}, and the fact that all the answer types in a handler are the same signifies that $\vdashst$ lacks \emph{answer-type modification} (ATM).  We shall show in Section~\ref{sec:mainresults} a type system that has ATM, $\vdashatm$, and show that the feature makes reachability decidable.  

\begin{figure}[t]
  \[
\begin{array}{c}
  \infer[(\textsc{St-Unit})]{\Gamma \vdashst () : \unitty}{} \ \
  \infer[(\textsc{St-Bool})]{\Gamma \vdashst v : \boolty}{v \in \{\tttrue,\ttfalse\}} \ \
  \infer[(\textsc{St-Var})]{\Gamma \vdashst x : \sigma}{x:\sigma \in \Gamma} \\\\
  \infer[(\textsc{St-Lam})]{\Gamma \vdashst \lambda x.c : \sigma \rightarrow \sigma'}{\Gamma,x:\sigma \vdashst c : \sigma'} \ \
  \infer[(\textsc{St-Rec})]{\Gamma \vdashst \ttrec{x}{v} : \sigma}{\Gamma, x:\sigma \vdashst v : \sigma} \\\\
  \infer[(\textsc{St-If})]{\Gamma \vdashst \ttif{v}{c_1}{c_2} : \sigma}{\Gamma \vdashst v : \boolty & \Gamma \vdashst c_1 : \sigma & \Gamma \vdashst c_2 : \sigma} \\\\
  \infer[(\textsc{St-App})]{\Gamma \vdashst v_1\;v_2 : \sigma'}{\Gamma \vdashst v_1 : \sigma\rightarrow\sigma' & \Gamma \vdashst v_2 : \sigma} \ \
  \infer[(\textsc{St-Let})]{\Gamma \vdashst \ttlet{x}{c_1}{c_2} : \sigma'}{\Gamma \vdashst c_1 : \sigma & \Gamma,x:\sigma \vdashst c_2 : \sigma'} \\\\
  \infer[(\textsc{St-Ret})]{\Gamma \vdashst \ttreturn\;v : \sigma}{\Gamma \vdashst v : \sigma} \ \
  \infer[(\textsc{St-Op})]{\Gamma \vdashst \itop\;v : \sigma'}{\Sigma(\itop) = \sigma \rightarrow\sigma' & \Gamma \vdashst v : \sigma}\\\\
  \infer[(\textsc{St-Hdlr})]{\Gamma \vdashst \{ \ttreturn\;x = c, \overline{\itop_i\;x_i\;k_i = c_i} \}: \sigma\rightarrow\sigma'}{\Gamma,x:\sigma \vdashst : c : \sigma' & \overline{\Gamma, x_i : \sigma_i, k_i : \sigma_i'\rightarrow \sigma' \vdashst c_i : \sigma'} & \overline{\Sigma(\itop_i) = \sigma_i\rightarrow \sigma_i'}}\\\\
  \infer[(\textsc{St-Han})]{\Gamma \vdashst \ttwithhandle{h}{c} : \sigma'}{\Gamma \vdashst h : \sigma \rightarrow \sigma' & \Gamma \vdashst c : \sigma}
   \end{array}
\]
\caption{The typing rules of the simple type system $\vdashst$.}
\label{fig:stypingrules}
\end{figure}

\begin{example}
\label{ex:sttypings}
Recall $c_{\it ex1}$ and $c_{\it ex2}$ from Examples~\ref{ex:loop} and \ref{ex:nondet}.  Both programs are $\vdashst$-typable.  For $c_{\it ex1}$, the types of some subexpressions are as follows:
\[
\setlength{\arraycolsep}{1em}
\begin{array}{ccc}
  v_{\it tup} : \boolty\rightarrow\boolty\rightarrow\sigma_{\it t} & v_{\it fst},v_{\it snd} : \sigma_{\it t} \rightarrow \boolty & c_{\it inc} : \unitty
\end{array}
\]
  where $\sigma_{\it t} = (\boolty \rightarrow \boolty \rightarrow \boolty) \rightarrow \boolty$.  And, the typing for the handler can be given by $x: \sigma_{\it t} \vdashst \lambda s.x : \sigma_{\it a}$
  for the return clause, and $x: \sigma_{\it t}, k: \unitty\rightarrow\sigma_{\it a} \vdashst \lambda s.k\;()\;x : \sigma_{\it a}$ and $x: \unitty, k: \sigma_{\it t}\rightarrow\sigma_{\it a} \vdashst \lambda s.k\;s\;s : \sigma_{\it a}$
  for the clauses of $\mathit{set}$ and $\mathit{get}$ respectively, where $\sigma_{\it a} = \sigma_{\it t}\rightarrow\unitty$.  Note that the answer type is $\sigma_{\it a}$ in all clauses, indicating that ATM is not needed to type the example.
  %
  Similarly, $c_{\it ex2}$ can be typed by typing the return clause as
  $x : \boolty \vdashst x: \boolty$, and the $\mathit{dec}$ clause as
  $x : \unitty, k : \boolty\rightarrow\boolty \vdashst v_{\it or}\;(k\;\top)\;(k\;\bot) : \boolty$.  Note that the answer type is $\boolty$ in all clauses in this typing, again indicating that ATM is not needed for typing the example.
  
\end{example}

We say that a $\plname$ program $c$ is $\vdashst$-\emph{typable} if $\vdashst c : \sigma$ for some $\sigma$.
We define untyped finitary PCF (without effect handlers) as the fragment of $\plname$ without $\ttwithhandle{h}{c}$ or $\itop\;v$, and we refer to the fragment by $\plnamenoeff$.  As mentioned in the introduction, the reachability problem for $\vdashst$-typable $\plnamenoeff$ is decidable~\cite{DBLP:conf/lics/Ong06,DBLP:conf/popl/Kobayashi09,DBLP:conf/fossacs/Tsukada014}.
\begin{theorem}[\cite{DBLP:conf/lics/Ong06,DBLP:conf/popl/Kobayashi09,DBLP:conf/fossacs/Tsukada014}]
  \label{thm:noeffdecidable}
Reachability is decidable for $\vdashst$-typable $\plnamenoeff$.
\end{theorem}
By contrast, as also mentioned in the introduction, a recent paper~\cite{DBLP:journals/pacmpl/LagoG24} has shown that the reachability problem is undecidable for $\vdashst$-typable $\plname$.\footnote{An alternative proof is given in \cite{DBLP:journals/pacmpl/Kobayashi25}.}
\begin{theorem}[\cite{DBLP:journals/pacmpl/LagoG24,DBLP:journals/pacmpl/Kobayashi25}]
Reachability is undecidable for $\vdashst$-typable $\plname$.
\end{theorem}


\section{Main Results}
\label{sec:mainresults}
\begin{wrapfigure}[5]{r}[0pt]{0.35\textwidth}
  \vspace{-30pt}
\[
\begin{array}{rcl}
  \tau & \mathrm{::=} & b \mid \tau \rightarrow \rho\\
  \rho &  \mathrm{::=} & \tau/\pureeff \mid \tau/\rho_1\Rightarrow\rho_2
\end{array}
\]
\vspace{-20pt}
\caption{The ATM types.}
\label{fig:atmtype}
\end{wrapfigure}
Figure~\ref{fig:atmtype} shows the ATM types.  The base types, $b$, remain unchanged from those of ordinary simple types. A \emph{value type}, denoted by $\tau$, is either a base-type or a function type of the form $\tau' \rightarrow \rho$ where $\rho$ is a \emph{computation type}.  A computation type is either of the form $\tau/\pureeff$ expressing a \emph{pure} computation that returns a value of type $\tau$, or of the form $\tau/\rho_1\Rightarrow\rho_2$ expressing an \emph{effectful} computation that changes the answer type from $\rho_1$ to $\rho_2$ and returns a value of  type $\tau$.

\begin{figure}[t]
\[
\begin{array}{c}
\infer[(\textsc{T-LetP})]
      {\Gamma \vdashatm \ttlet{x}{c_1}{c_2} : \tau_2/\pureeff}
      {\Gamma \vdashatm c_1 : \tau_1/\pureeff & \Gamma,x:\tau_1 \vdashatm c_2 : \tau_2/\pureeff}\\\\
\infer[(\textsc{T-LetIp})]
      {\Gamma \vdashatm \ttlet{x}{c_1}{c_2} : \tau_2/\rho_2\Rightarrow\rho_1'}
      {\Gamma \vdashatm c_1 : \tau_1/\rho_1\Rightarrow\rho_1' & \Gamma,x:\tau_1 \vdashatm c_2 : \tau_2/\rho_2\Rightarrow\rho_1}\\\\
\infer[(\textsc{T-Ret})]
      {\Gamma \vdashatm \ttreturn\;v : \tau/\pureeff}
      {\Gamma \vdashatm v : \tau} \ \
\infer[(\textsc{T-Op})]
      {\Gamma \vdashatm \itop\;v : \tau'/\rho_1\Rightarrow\rho_2}
      {\Sigma(\itop) = \tau\rightarrow\tau'/\rho_1\Rightarrow\rho_2 & \Gamma \vdashatm v : \tau}\\\\
\infer[(\textsc{T-Hdlr)}]
      {\Gamma \vdashatm \{ \ttreturn\;x = c, \overline{\itop_i\;x_i\;k_i = c_i}\}}
      {\overline{\Sigma(\itop_i) = \tau_i\rightarrow\tau_i'/\rho_i\Rightarrow\rho_i'} & \overline{\Gamma, x_i:\tau_i, k_i:\tau_i'\rightarrow\rho_i \vdashatm c_i : \rho_i'}}\\\\
\infer[(\textsc{T-Han})]
      {\Gamma \vdashatm \ttwithhandle{h}{c} : \rho'}
      {\Gamma \vdashatm h & \Gamma \vdashatm c : \tau/\rho\Rightarrow\rho' & \Gamma,x:\tau \vdashatm c' : \rho & \ttreturn\;x = c' \in h}\\\\
\infer[(\textsc{T-VSub})]
      {\Gamma \vdashatm v : \tau'}
      {\Gamma \vdashatm v : \tau & \tau \leq \tau'}\ \
\infer[(\textsc{T-CSub})]
      {\Gamma \vdashatm c : \rho'}
      {\Gamma \vdashatm c : \rho & \rho \leq \rho'} 
\end{array}
\]
\caption{Representative typing rules of the ATM type system $\vdashatm$.}
\label{fig:atmtypingrules}
\end{figure}

Figure~\ref{fig:atmtypingrules} shows representative typing rules of the ATM type system $\vdashatm$.  We refer to \full{Appendix~\ref{app:atmtypingrulesfull}}{the extended report~\cite{fullversion}} for the complete set.  The type system is essentially the ATM refinement type system proposed in \cite{DBLP:journals/pacmpl/KawamataUST24}, but without the refinement types aspect that is orthogonal to our paper.
(\textsc{T-LetP}) stipulates that if both $c_1$ and $c_2$ are pure computations then the entire computation is also a pure one.  By contrast, (\textsc{T-LetIp}) says that if $c_1$ is an effectful computation that changes the answer type from $\rho_1$ to $\rho_1'$ and $c_2$ is an effectful computation that changes the answer type from $\rho_2$ to $\rho_1$, then the entire computation is an effectful computation that changes the answer type from $\rho_2$ to $\rho_1'$.  Note that the answer types of the sub-computations are composed in a backward manner (cf.~\cite{DBLP:journals/pacmpl/KawamataUST24} for the explanation).  (\textsc{T-Ret}) and (\textsc{T-Op}) are analogous to the corresponding rules (\textsc{St-Ret}) and (\textsc{St-Op}) of $\vdashst$.  In particular, the latter looks up the type of the operation in the signature $\Sigma$ whose return computation type $\tau'/\rho_1\Rightarrow\rho_2$ is an effectful one that changes the answer type from $\rho_1$ to $\rho_2$.  (\textsc{T-Hdlr}) checks that a handler is well-typed.  Note that, unlike (\textsc{St-Hdlr}) of $\vdashst$, the rule allows the operation clause bodies and their continuation parameters to have different answer types, signifying that $\vdashatm$ allows ATM.  The return clause is typed in (\textsc{T-Han}) and it too is allowed to have a different answer type.  Additionally, (\textsc{T-Han}) stipulates that the type of the with-handle block is changed from that of the return clause body, $\rho$, to $\rho'$ by ATM.  The last two rules, (\textsc{T-VSub}) and (\textsc{T-CSub}), are subsumption rules for the subtyping relation $\leq$.

A key subtyping rule is the following one that allows ``embedding'' a pure computation type into an effectful computation type:
\[
\infer[(\textsc{S-Embed})]
      {\tau_1/\pureeff \leq \tau_2/\rho_1\Rightarrow\rho_2}
      {\tau_1 \leq \tau_2 & \rho_1 \leq \rho_2}
\]
The rule signifies that a pure computation can always be used where an effectful one is expected.  The remaining subtyping rules are defined inductively on the structure of the types.  We refer to \full{Appendix~\ref{app:atmsubtypingrulesfull}}{the extended report~\cite{fullversion}} for the complete set of subtyping rules.
We say that a program $c$ is \emph{$\vdashatm$-typable} if $\vdashatm c : \tau/\pureeff$ for some $\tau$.\footnote{The restriction to pure types for $\vdashatm$-typable programs is for simplicity.  It lets us disregard the case the source program gets stuck with an unhandled operation when showing the correctness of the CPS transformation.  It can be shown that reachability remains decidable for $\vdashatm$-typable programs even without the restriction.}


\begin{example}
\label{ex:atmtypings}
Recall $c_{\it ex1}$ and $c_{\it ex2}$ from Examples~\ref{ex:loop} and \ref{ex:nondet}.  Recall that both programs are $\vdashst$-typable as shown in Example~\ref{ex:sttypings}.  We show that both programs are also $\vdashatm$-typable.  For $c_{\it ex1}$, the typing for the handler can be given by $x:\tau_t \vdashatm \lambda s.x:\rho_a$ for the return clause, $x:\sigma_t,k:\unitty\rightarrow \rho_a \vdashatm \lambda s. k\;()\;x:\rho_a$ for the $\mathit{set}$ clause, and $x:\unitty,k:\tau_t\rightarrow\rho_a \vdashatm \lambda x.k\;s\;s:\rho_a$ for the $\mathit{get}$ clause, where
\[
\begin{array}{rcl}
  \tau_t & = & (\boolty\rightarrow(\boolty\rightarrow\boolty/\pureeff)/\pureeff) \rightarrow \boolty/\pureeff\\
  \rho_a & = & (\tau_t\rightarrow\unitty/\pureeff)/\pureeff
\end{array}
\]
Typing these clauses does not need (\textsc{T-LetIp}) but only (\textsc{T-LetP}) because the computation in the clauses are all pure.\footnote{(\textsc{T-LetP}) is used when the trivial let bindings hidden by the notational convention are expanded (cf.~Section~\ref{sec:prelim}).}  By contrast, the operation invocations of $\mathit{get}$ and $\mathit{set}$ will respectively be given effectful computation types $\tau_t/\rho_a\Rightarrow\rho_a$ and $\unitty/\rho_a\Rightarrow\rho_a$ by (\textsc{T-OP}).  The body of the with-handle block will also be given an effectful computation type $\unitty/\rho_a\Rightarrow\rho_a$ by using (\textsc{T-LetIp}) to compose effectful computation types and giving the recursive function $\mathit{loop}$ the type $\unitty\rightarrow\unitty/\rho_a\Rightarrow\rho_a$.  Therefore, $\vdashatm c_{\it ex1} : \boolty/\pureeff$.
Similarly, $c_{\it ex2}$ can be $\vdashatm$-typed by typing the return clause and the $\mathit{dec}$ clause as $x:\boolty \vdashatm x : \rho_b$ and
$x : \unitty, k:\boolty\rightarrow\rho_b \vdashatm v_{\it or}\;(k\;\tttrue)\;(k\;\ttfalse) : \rho_b$, respectively, where $\rho_b = \boolty/\pureeff$.  The invocations of $\mathit{dec}$ in the with-handle-block body can be given the type $\boolty/\rho_b\Rightarrow\rho_b$ and so can the body itself.
\end{example}

\begin{example}
\label{ex:stuntypableatmtypable}
  The previous example programs $c_{\it ex1}$ and $c_{\it ex2}$ were both $\vdashst$-typable and $\vdashatm$-typable.  Here, we show an example of a program that \emph{needs} ATM to be typed.   That is, the program is $\vdashatm$-typable but not $\vdashst$-typable.  Consider the following program $c_{\it ex3}$.
  \[
  \begin{array}{l}
  \mathtt{let}\;r = {\mathtt{with}\;\{ \ttreturn\;x = x, \itop\;x\;k = k\;x;\tttrue \}\;\mathtt{handle}\; (\itop\;();())}\;\mathtt{in}\\
  \ttif{r}{\tttrue}{\ttfalse}
  \end{array}
  \]
  Note that the program is semantically type safe.  In particular, note that the value that gets bound to $r$ is the Boolean value $\top$ rather than unit.
  The program is $\vdashatm$-typable by giving the invocation of $\itop$ the type $\unitty/(\unitty/\pureeff \Rightarrow \boolty/\pureeff)$
  and giving the variable $r$ the type $\boolty$.  However, it is not $\vdashst$-typable because, in that type system, the return clause body needs to be given the type $\unitty$ whereas the $\itop$ clause body needs to be given the type $\boolty$ and the rule for typing a handler (\textsc{St-Hdlr}) asserts that these types need to be the same.
\end{example}

As shown in the above example, there are programs typable with ATM but not without it.  In fact, one may naturally expect allowing ATM makes more programs typable because it allows an operation invocation to change the type of a with-handle block from that of the return clause body to that of an operation clause body.  Therefore, the following main result of the paper may come as a surprise because, as remarked before, reachability is undecidable for $\vdashst$-typable $\plname$.
\begin{theorem}[Decidability]
\label{thm:main}
Reachability is decidable for $\vdashatm$-typable $\plname$.
\end{theorem}

The rest of this section is organized as follows.  We prove Theorem~\ref{thm:main} in Section~\ref{subsec:cps} by presenting a novel typing-derivation-directed CPS transformation that transforms an $\vdashatm$-typable $\plname$ program to a $\vdashst$-typable $\plnamenoeff$ program.  A corollary of Theorem~\ref{thm:main} is the existence of a $\plname$ program that is $\vdashst$-typable but not $\vdashatm$-typable, and we present a concrete example of such a program in Section~\ref{subsec:atmuntypableex} (Section~\ref{subsec:termination} also has such an example).  Section~\ref{subsec:termination} describes another application of our CPS transformation.  That is, we show there that every $\vdashatm$-typable recursive-definition-free $\plname$ program terminates while the same does not hold for $\vdashst$-typable ones.  Section~\ref{subsec:errorinsekiyamaunno} disproves a claim made in a recent paper~\cite{DBLP:journals/pacmpl/Sekiyama024} regarding what the paper calls \emph{active effect handlers}.

\subsection{Typing-Derivation-Directed CPS Transformation and the Proof of Theorem~\ref{thm:main}}
\label{subsec:cps}

We prove Theorem~\ref{thm:main} by presenting a CPS transformation that transforms an $\vdashatm$-typable $\plname$ program to a $\vdashst$-typable $\plnamenoeff$ program.  Then, Theorem~\ref{thm:main} follows from the fact that reachability is decidable for $\vdashst$-typable $\plnamenoeff$ programs.
%

We note that a CPS transformation for a language with effect handlers and ATM has already been proposed in a recent work by Kawamata et al.~\cite{DBLP:journals/pacmpl/KawamataUST24}.  However, their CPS transformation requires higher-rank parametric polymorphism in an essential way to type the target of the transformation, and therefore, it is insufficient for our purpose because reachability for $\plnamenoeff$ programs typable with a higher-rank parametric polymorphic type system is undecidable~\cite{DBLP:conf/fossacs/TsukadaK10}.

Our key observation is that the CPS transformation of \cite{DBLP:journals/pacmpl/KawamataUST24} uses parametric polymorphism to allow a pure computation to be used in contexts where an effectful computation is expected.  This is what subtyping is used for in $\vdashatm$.  Based on the observation, and inspired by the CPS transformation proposed in a paper by Materzok and Biernacki~\cite{DBLP:conf/icfp/MaterzokB11}, we propose a new CPS transformation for effect handlers that is subtyping-aware so that the target of the transformation can be typed without parametric polymorphism.  The key idea, like that of \cite{DBLP:conf/icfp/MaterzokB11}, is to make the transformation be directed by the \emph{typing derivation} of the source expression, instead of being only \emph{type}-directed as more commonly seen for a CPS transformation.  

For convenience, we extend $\vdashst$-typed $\plnamenoeff$ with records.  That is, we extend the syntax of expressions and simple types as follows.
\[
\setlength\arraycolsep{0.4cm}
\begin{array}{ccc}
  v\;\mathrm{::=}\;\dots \mid \{\overline{l_i = v_i} \} &
  c\;\mathrm{::=}\;\dots \mid v.l &
  \sigma\;\mathrm{::=}\;\dots \mid \{ \overline{l_i:\sigma_i} \}
\end{array}
\]
The extension to the operational semantics and the typing rules is standard and is given in \full{Appendix~\ref{app:recordext}}{the extended report~\cite{fullversion}}.  We note that the reachability for $\plnamenoeff$ remains decidable with the record extension.\footnote{This can be shown, for example, by encoding records as functions similarly to how it was done for tuples in Example~\ref{ex:loop}.}

\begin{figure}[t]
\[
\begin{array}{l}
  {\sembrack{b} = b \ \ \ \
  \sembrack{\tau\rightarrow\rho} = \sembrack{\tau}\rightarrow\sembrack{\rho} \ \ \ \
  \sembrack{\tau/\pureeff} = \sembrack{\tau}} \\
  \sembrack{\tau/\rho_1\Rightarrow\rho_2} = \sembrack{\Sigma}\rightarrow(\sembrack{\tau}\rightarrow\sembrack{\rho_1})\rightarrow\sembrack{\rho_2} \\
  \sembrack{\{ \overline{\itop_i : \tau_i\rightarrow\tau_i'/\rho_i\Rightarrow\rho_i'}} = \{ \overline{ \itop_i : \sembrack{\tau_i}\rightarrow (\sembrack{\tau_i'}\rightarrow \sembrack{\rho_i} ) \rightarrow \sembrack{\rho_i'}} \}
\end{array}
\]
\caption{The CPS transformation of types.}
\label{fig:cpstypes}
\end{figure}

Figure~\ref{fig:cpstypes} shows the CPS transformation of types.  As seen in the second to the last rule, an effectful computation is transformed to a function that takes a transformed signature (i.e., a record of the type $\sembrack{\Sigma}$), a transformed continuation of the type $\sembrack{\tau} \rightarrow \sembrack{\rho_1}$, and returns a value of the type $\sembrack{\rho_2}$.
For a type environment $\Gamma$, we denote by $\sembrack{\Gamma}$ the CPS transformed environment $\{ x:\sembrack{\tau} \mid x:\tau \in \Gamma \}$.

Next, we define the CPS transformation of subtyping derivations, which is of the form $\sembrack{\tau_1 \leq \tau_2}$ for subtyping of value types and $\sembrack{\rho_1 \leq \rho_2}$ for subtyping of computation types.  An important case is the transformation of a derivation whose root is an instance of (\textsc{S-Embed}) shown below.
  \[
 \sembrack{\tau_1/\pureeff \leq \tau_2/\rho_1\Rightarrow\rho_2} = \lambda x.\lambda h.\lambda k.\sembrack{\rho_1 \leq \rho_2}@(k@(\sembrack{\tau_1 \leq \tau_2}@x)
 \]
where $@$ denotes a \emph{static application} that is processed during the CPS transformation~\cite{DBLP:conf/rta/HillerstromLAS17,DBLP:journals/pacmpl/KawamataUST24}. The complete set of the CPS transformation rules that concern subtyping derivations is given in \full{Appendix~\ref{app:cpssubtyping}}{the extended report~\cite{fullversion}}.

Finally, we define the CPS transformation of typing derivations, which is of the form $\sembrack{\Gamma \vdashatm v : \tau}$ for value expression typing and $\sembrack{\Gamma \vdashatm c: \rho}$ for computation expression typing.  An interesting rule is one concerning the subsumption rule (\textsc{T-CSub}) shown below.
\[
\sembrackstretch{\dfrac{\Gamma \vdashatm c : \rho \ \ \rho \leq \rho'}{\Gamma \vdashatm c : \rho'}} = \sembrackstretch{\rho \leq \rho'}@\sembrackstretch{\Gamma \vdashatm c : \rho}
\]
Note that it uses the transformation obtained from the subtyping $\rho \leq \rho'$ to properly CPS transform the source computation expression $c$ that has been CPS transformed with respect to the sub-derivation $\Gamma \vdashatm c : \rho$.  Another exemplifying transformation rule is one concerning the (\textsc{T-Hdlr}) rule for typing a handler shown below.
\[
\begin{array}{c}
  \left\llbracket \dfrac{\overline{\Sigma(\itop_i) = \tau_i\rightarrow\tau_i'/\rho_i\Rightarrow\rho_i'} \ \ \overline{\Gamma, x_i:\tau_i, k_i:\tau_i'\rightarrow\rho_i \vdashatm c_i : \rho_i'}}{\Gamma \vdashatm \{ \ttreturn\;x = c, \overline{\itop_i\;x_i\;k_i = c_i}\}} \right\rrbracket \\
  \hspace{3.5cm} = \left\{ \overline{\itop_i = \lambda x_i.\lambda k_i.\sembrack{\Gamma, x_i:\tau_i, k_i : \tau_i'\rightarrow\rho_i \vdashatm c_i : \rho_i'}} \right\}
  \end{array}
\]
Note that the transformed result is a record mapping each operation name $\itop_i$ to a function obtained by transforming the typing derivation for the operation clause body $c_i$.  The rule is used in the CPS transformation corresponding to the typing rule (\textsc{T-Han}) for typing a with-handle block shown below.
\[
\begin{array}{c}
  \sembrackstretch{\dfrac{\Gamma \vdashatm h \ \ \Gamma \vdashatm c : \tau/\rho\Rightarrow\rho' \ \ \Gamma,x:\tau \vdashatm c' : \rho \ \ \ttreturn\;x = c' \in h}{\Gamma \vdashatm \ttwithhandle{h}{c} : \rho'}}\\
  \hspace{2.5cm} = \sembrackstretch{\Gamma \vdashatm c : \tau/\rho\Rightarrow\rho'}@\sembrackstretch{\Gamma \vdashatm h}@\lambda x.\sembrackstretch{\Gamma,x:\tau \vdashatm c' : \rho}
  \end{array}
\]
The rule uses the transformation rule corresponding to (\textsc{T-Hdlr}) mentioned above to transform the handler typing derivation $\Gamma \vdashatm h$ to a record of operations, transform the return clause typing derivation $\Gamma,x:\tau \vdashatm c' : \rho$, and apply the transformation of the typing derivation for the with-handle block body $\Gamma \vdashatm c : \tau/\rho\Rightarrow\rho'$, which would be a function of the type $\sembrack{\Sigma}\rightarrow(\sembrack{\tau}\rightarrow\sembrack{\rho})\rightarrow\sembrack{\rho'}$, to the former and the $\lambda$ abstraction of the latter.  The passed record will be looked up in the transformation of an operation invocation, as seen in the following transformation rule corresponding to (\textsc{T-Op}).
\[
\begin{array}{c}
  \sembrackstretch{\dfrac{\Sigma(\itop) = \tau\rightarrow\tau'/\rho_1\Rightarrow\rho_2 \ \ \Gamma \vdashatm v : \tau}{\Gamma \vdashatm \itop\;v : \tau'/\rho_1\Rightarrow\rho_2}} = \lambda h.\lambda k. h.\itop\;\sembrack{\Gamma \vdashatm v : \tau}\;k
\end{array}
\]
We refer to \full{Appendix~\ref{app:cpstyping}}{the extended report~\cite{fullversion}} for the complete set of transformation rules that concern typing derivations.

We show the correctness of our CPS transformation.  First, we show that the transformed program is a $\vdashst$-typable $\plnamenoeff$ program, which follows immediately from the following theorem that can be proven by induction on (sub)typing derivations, and the fact that the right hands of the transformation rules do not contain effect handlers.
\begin{theorem}[Typability Preservation]
\label{thm:typabilitypreservation}
The following holds.
\begin{enumerate}
  \item If $\tau \leq \tau'$ then $\vdashst \sembrack{\tau \leq \tau'} : \sembrack{\tau}\rightarrow\sembrack{\tau'}$.
  \item If $\rho \leq \rho'$ then $\vdashst \sembrack{\rho \leq \rho'} : \sembrack{\rho}\rightarrow\sembrack{\rho'}$.
  \item If $\Gamma \vdashatm v : \tau$ then $\sembrack{\Gamma} \vdashst \sembrack{\Gamma \vdashatm v : \tau} : \sembrack{\tau}$.
  \item If $\Gamma \vdashatm c : \rho$ then $\sembrack{\Gamma} \vdashst \sembrack{\Gamma \vdashatm c : \rho} : \sembrack{\rho}$.
\end{enumerate}
\end{theorem}
Next, we show that the transformation preserves reachability.  That is, if $c$ is $\vdashatm$-typable by a derivation $\vdashatm c : \tau/\pureeff$, then the answer to the reachability problem for $c$ is the same as that for $\sembrack{\;\vdashatm c: \tau/\pureeff }$.
This is shown by the following simulation theorem.  Let $\rightarrow^+$ denote the transitive closure of $\rightarrow$.
\begin{theorem}[Simulation]
\label{thm:simulation}
  Suppose $\vdashatm c : \tau/\pureeff$.  The following holds.
  \begin{enumerate}
  \item \label{item:forward} If $c \rightarrow^* \ttreturn\;v$ then $\vdashatm v : \tau$ and $\sembrack{\;\vdashatm c : \tau/\pureeff} \rightarrow^+ \sembrack{\;\vdashatm v : \tau}$.
  \item \label{item:backward} If $\sembrack{\;\vdashatm c : \tau/\pureeff} \rightarrow^+ \ttreturn\;v'$ then there is $v$ such that $\vdashatm v : \tau$, $\sembrack{\;\vdashatm v : \tau} = v'$, and $c \rightarrow^* \ttreturn\;v$.
    \end{enumerate}
\end{theorem}
We refer to \full{Appendix~\ref{app:simulation}}{the extended report~\cite{fullversion}} for the proof.
Therefore, given an $\vdashatm$-typable $c$, we can decide the reachability of $c$ by deciding the reachability of $\sembrack{\;\vdashatm c : \tau/\pureeff}$, because $\sembrack{\;\vdashatm c : \tau/\pureeff}$ is a $\vdashst$-typable $\plnamenoeff$ program by Theorem~\ref{thm:typabilitypreservation} and the reachability problem of $\vdashst$-typable $\plnamenoeff$ is decidable.  Thus, reachability for $\vdashatm$-typable $\plname$ is decidable.  This completes the proof of Theorem~\ref{thm:main}.

\subsection{A Concrete Example of a $\vdashst$-Typable but $\vdashatm$-Untypable Program}
\label{subsec:atmuntypableex}

A corollary of our decidability result (Theorem~\ref{thm:main}) is that there are $\vdashst$-typable programs that are not $\vdashatm$-typable, which may seem counter-intuitive because one may naturally think that allowing ATM can only make more programs typable. We give a concrete example of such a program.  Let 
\[
\begin{array}{rcl}
  c_{\it ex4}  & \defeq  & \ttrec{f}{\lambda z.\ttwithhandle{h}{f\;()}}\\
  h  & \defeq & \{ \ttreturn\;x = \itop{}\;(), \itop\;x\;k = () \}
\end{array}
\]
This program $c_{\it ex4}$ is $\vdashst$-typable.  Namely, $\vdashst c_{\it ex4} : \unitty\rightarrow\unitty$ with $\Sigma(\itop) = \unitty\rightarrow\unitty$.  The recursively defined function $f$ would also be given the type $\unitty\rightarrow\unitty$.
We now show that $c_{\it ex4}$ is $\vdashatm$-untypable.  Suppose for contradiction that it is $\vdashatm$-typable.  It must be the case that the recursive function $f$ is given the type $\unitty\rightarrow\rho_f$ for some $\rho_f$.  Because the body of the with-handle block is $f\;()$, by (\textsc{T-Han}), this $\rho_f$ must satisfy $\rho_f \leq \tau/\rho_1\Rightarrow\rho_2$ for some $\tau$, $\rho_1$, and $\rho_2$.  However, by (\textsc{T-Han}) again, $\rho_2$ must be a subtype of the with-handle block and thus a subtype of $\rho_f$ by (\textsc{T-Rec}).  That is, we have the subtyping relation $\rho_2 \leq \tau/\rho_1 \Rightarrow \rho_2$.
We state the following general property of ATM subtyping, which will be used in our argument.
\begin{lemma}
If $\rho \leq \tau/\rho'\Rightarrow\rho$ then $\rho'$ is pure.
\end{lemma}
\begin{proof}
By induction on the structure of $\rho$.\qed
\end{proof}
Therefore, $\rho_1$ must be pure.  This is not possible because $\rho_1$ must be a supertype of the type of the return clause, but the return clause invokes the operation $\itop$ and therefore must be given an effectful type.  Thus, $c_{\it ex4}$ is not $\vdashatm$-typable.

\subsection{Termination and Non-Termination for $\vdashatm$-Typable and $\vdashst$-Typable Programs}
\label{subsec:termination}

The example of a $\vdashst$-typable but $\vdashatm$-untypable program given in Section~\ref{subsec:atmuntypableex} used a recursive definition in an essential way.  A natural question is whether this is always the case, that is, recursive definitions are necessary for a program to be $\vdashst$-typable but not $\vdashatm$-typable.  We answer the question negatively by presenting a recursive-definition-free program that is $\vdashst$-typable but not $\vdashatm$-typable.
We do this by presenting a result that may be of independent interest and says that any $\vdashatm$-typable recursive-definition-free program terminates whereas the same is not true for the $\vdashst$-typable ones.

Consider the following recursive-definition-free program $c_{\it ex5}$:
\[
(\mathtt{with}\;h_{\it state}\;\mathtt{handle}\;\mathtt{let}\;f = \lambda x.\mathit{get}\;()\;()\;\mathtt{in}\;(\mathit{set}\;f;f\;()))\;\lambda x.()
\]
where $h_{\it state}$ is the mutable state handler from Example~\ref{ex:loop}.
This program essentially implements the textbook encoding of an infinite loop by a mutable state (i.e., Landin's knot), and is non-terminating.  Indeed, we have
\[
\begin{array}{rcl}
  c_{\it ex5} & \rightarrow^+ & (\lambda s.(\lambda y.\ttwithhandle{h_{\it state}}{y;v_f\;()})\;()\;v_f)\;\lambda x.() \\
  & \rightarrow^+ & (\lambda y.\ttwithhandle{h_{\it state}}{y;v_f\;()})\;()\;v_f \\
  & \rightarrow^+ & (\ttwithhandle{h_{\it state}}{v_f\;()})\;v_f \\
  & \rightarrow^+ & (\ttwithhandle{h_{\it state}}{\mathit{get}\;()\;()})\;v_f\\
  & \rightarrow^+ & (\lambda s.(\lambda x.\ttwithhandle{h_{\it state}}{x\;()})\;s\;s)\;v_f\\
  & \rightarrow^+ &(\lambda x.\ttwithhandle{h_{\it state}}{x\;()})\;v_f\;v_f \\
  & \rightarrow^+ & (\ttwithhandle{h_{\it state}}{v_f\;()})\;v_f
\end{array}
\]
where $v_f = \lambda\_.\mathit{get}\:()\:()$.  Because the last line is the same as the third line, the evaluation diverges.
The program $c_{\it ex5}$ is $\vdashst$-typable.  Namely, it can be $\vdashst$-typed by giving $\mathit{get}$ the type $\unitty\rightarrow\sigma_t$ and $\mathit{set}$ the type $\sigma_t\rightarrow\unitty$ in the operation signature, where $\sigma_t = \unitty\rightarrow\unitty$.  The typing of the return, $\mathit{get}$, and $\mathit{set}$ clauses can be done in the same way as in Example~\ref{ex:sttypings} except for using the above $\sigma_t$ for the $\sigma_t$ there.
The with-handle block can be given the type $\unitty$ with these types of
$\mathit{get}$ and $\mathit{set}$.
However, $c_{\it ex5}$ is not $\vdashatm$-typable.  This follows from the theorem below which can be easily shown by using our CPS transformation.
\begin{theorem}
Every $\vdashatm$-typable $\plname$ program without recursive definitions terminates.
\end{theorem}
\begin{proof}
  Let $c$ be an $\vdashatm$-typable $\plname$ program without recursive definitions, and let $c' = \sembrack{\;\vdashatm c : \tau/\pureeff}$.  Because our CPS transformation does not add new recursive definitions, $c'$ also does not contain recursive definitions (in addition to not containing effect handlers).  By the results shown in Section~\ref{subsec:cps}, $c'$ is $\vdashst$-typable and $c'$ terminates iff $c$ terminates.  Because simply-typed $\lambda$ calculus (without recursive definitions or effect handlers) is terminating, $c'$ must terminate and therefore so must $c$. \qed
\end{proof}

\subsection{Number of Active Effect Handlers Does Not Capture Decidability}
\label{subsec:errorinsekiyamaunno}

A recent paper~\cite{DBLP:journals/pacmpl/Sekiyama024} introduced a notion called \emph{active effect handlers}.  Their paper claims that the boundedness of the number of active effect handlers characterizes the decidability of reachability, and that the reason why ATM ensures the decidability is because it ensures that this number is bounded.
In this section, we disprove this claim by presenting a class of programs with only a bounded number of active effect handlers (in fact, with only at most one active effect handler) but whose reachability is nonetheless undecidable.

Let $\mathit{succ}$ be an operation, and $x_0$, $x_1$, $f_0,\dots,f_n$ be variables where $n \geq 0$.  Let us define the set of computation expressions $\mathit{Inst}^n \defeq \{ c_{\it inc}^{i,j} \mid i \in \{0,1\}, j \in \{0,\dots,n\} \} \cup \{ c_{\it dec}^{i,j,m} \mid i \in \{0,1\}, j,m \in \{0,\dots,n\} \cup \{()\}$ where
\[
\begin{array}{rcl}
  c_{\it inc}^{i,j} & \defeq & \ttlet{x_i}{\lambda y.\mathit{succ}\;x_i}{f_j\;x_0\;x_1}\\
  c_{\it dec}^{i,j,m} & \defeq & \ttwithhandle{\{\ttreturn\;x = f_j\;x_0\;x_1, \mathit{succ}\;x_i\;k = f_m\;x_0\;x_1 \}}{x_i\;()}
\end{array}
\]
Then, let $\mathit{MM}^n$ be the class of programs of the form
\[
  \mathtt{mrec}\;f_0 = \lambda x_0.\lambda x_1. c_0\;\mathtt{and}\; \dots
  \;\mathtt{and}\;f_n = \lambda x_0.\lambda x_1. c_n\;\mathtt{in}
  \;(f_0\;(\lambda x.())\;\lambda x.());\tttrue
\]
where each $c_i \in \mathit{Inst}^n$, and the mutual recursive definition $\mathtt{mrec}\;f_0=v_0\;\mathtt{and}\;\dots\;\\\mathtt{and}\;f_n=v_n\;\mathtt{in}\;c$ is syntactic sugar defined inductively by:
\[
\begin{array}{l}
  \mathtt{let}\;f_0 = \ttrec{f_0}{(\mathtt{mrec}\;f_1=v_1\;\mathtt{and}\;\dots\;\mathtt{and}\;f_n=v_n\;\mathtt{in}\;v_0)}\;\mathtt{in} \\
  \hspace{5cm}\vdots\\
  \mathtt{let}\;f_n = \ttrec{f_n}{(\mathtt{mrec}\;f_0=v_0\;\mathtt{and}\;\dots\;\mathtt{and}\;f_{n-1}=v_{n-1}\;\mathtt{in}\;v_n)}\;\mathtt{in}\;c
\end{array}
\]
Let $\mathit{MM} = \bigcup_{n \in \mathbb{N}} \mathit{MM}^n$.  We refer to \full{Appendix~\ref{app:minskymachine}}{the extended report~\cite{fullversion}} for the proof of the following theorem.
\begin{theorem}
\label{thm:minskymachine}
Reachability for $\mathit{MM}$ is undecidable.
\end{theorem}
We note that $\mathit{MM}$ is $\vdashst$-typable.  Indeed, any program in $\mathit{MM}$ can be $\vdashst$-typed by giving the operation $\mathit{succ}$ the type $\sigma_{\it nat} \rightarrow \unitty$ and each recursive function $f_i$ the type $\sigma_{\it nat}\rightarrow\sigma_{\it nat}\rightarrow\unitty$ where $\sigma_{\it nat} = \unitty\rightarrow\unitty$.

Next, we recall the notion of \emph{active effect handlers} from \cite{DBLP:journals/pacmpl/Sekiyama024}.  Concretely, the number of active effect handlers is said to be \emph{bounded} for a program if there is a non-negative integer $n$ such that the evaluation of the program only yields intermediate expressions of the form
\[
 \dots (\mathtt{with}\;h_1\;\mathtt{handle} \dots (\mathtt{with}\;h_m\;\mathtt{handle}\;c)\;\dots) \dots
 \]
for $m \leq n$ where $c$ does not contain a with-handle block.  Roughly, the number of active effect handlers measures the number of pending effect handlers on the call stack (when effect handlers are implemented by using a call stack, as often done in real implementations).

It is easy to see that the number of active effect handlers for any program in $\mathit{MM}$ is bounded (in fact, by one), because a with-handle block will only appear in the evaluation by a calling a recursive function whose body is some $c_{dec}$, but then a subsequent recursive function call can happen only by replacing this with-handle block by the body of the return clause or the $\mathit{succ}$ clause.
Because reachability for $\mathit{MM}$ is undecidable as shown in Theorem~\ref{thm:minskymachine}, this disproves the claim that the boundedness of the number of active effect handlers characterizes the decidability of reachability for finitary PCF with effect handlers.

Additionally, a recent paper~\cite{DBLP:journals/pacmpl/LagoG24} contains a result stating that the reachability problem becomes decidable when the operation clauses are restricted to only use the given delimited continuation at tail positions.  This may appear to contradict Theorem~\ref{thm:minskymachine} because the operation (i.e., $\mathit{succ}$) clause in $\mathit{MM}$ does not use the given delimited continuation.  But it actually does not, because \cite{DBLP:journals/pacmpl/LagoG24} restricts the (non-continuation) parameters of operations to base types and thus disallows programs like $\textit{MM}$.\footnote{The same restriction is used in \cite{DBLP:journals/pacmpl/Kobayashi25}.}  The main result of our paper (Theorem~\ref{thm:main}) shows that such restrictions on operation parameter types or delimited continuation usage are not need for ATM to ensure the decidability of reachability.


\section{Related Work}
\label{sec:related}
As mentioned in the introduction, our work is inspired by the prior research on reachability for finitary PCF (without effect handlers).  The problem was shown to be decidable~\cite{DBLP:conf/lics/Ong06,DBLP:conf/popl/Kobayashi09,DBLP:conf/fossacs/Tsukada014}, and the result has served as a foundation of methods for verifying infinite-state higher-order-recursive programs by incorporating techniques like predicate abstraction and CEGAR to abstract infinite data to finite domains~\cite{DBLP:conf/popl/BallR02,DBLP:conf/cav/ClarkeGJLV00,DBLP:conf/pldi/BallMMR01,DBLP:conf/popl/Terauchi10,DBLP:conf/popl/OngR11,DBLP:conf/pldi/KobayashiSU11,DBLP:conf/popl/UnnoTK13,DBLP:conf/popl/RamsayNO14}.  It is worth noting that such studies have extended beyond just reachability (i.e., safety properties) and have also lead to methods for verifying liveness properties by incorporating techniques like binary reachability analysis and automata-theoretic verification method~\cite{DBLP:journals/apal/Vardi91,DBLP:conf/lics/PodelskiR04,DBLP:conf/pldi/CookPR06,DBLP:conf/popl/CookGPRV07,DBLP:conf/esop/KuwaharaTU014,DBLP:conf/popl/MuraseT0SU16}.

Inspired by this success and the popularity of effect handlers, a recent paper by Dal Lago and Ghyselen~\cite{DBLP:journals/pacmpl/LagoG24} has investigated the decidability of reachability for finitary PCF extended with effect handlers, and has found that the problem is undecidable.  An alternative proof of this undecidability is also given in a recent paper by Kobayashi~\cite{DBLP:journals/pacmpl/Kobayashi25}.
In this paper, we have shown that this undecidability comes from the way the standard type systems for effect handlers are designed, and that, surprisingly, the problem becomes decidable when the type system is extended to allow ATM.  As remarked in Section~\ref{subsec:errorinsekiyamaunno}, \cite{DBLP:journals/pacmpl/LagoG24} contains a result stating that the reachability problem becomes decidable if the operation clauses are restricted to only use the captured delimited continuation at tail positions.  However, as we have shown there, this decidability result crucially relies on the fact that operation parameters are restricted to base types in \cite{DBLP:journals/pacmpl/LagoG24} because allowing arbitrary (simple) types for operation parameters makes reachability undecidable even with the restriction on the delimited continuation usage.
The decidability result of our paper shows that such restrictions on operation parameter types or delimited continuation usage are not needed for ATM to ensure the decidability of reachability.

Our decidability result was proven by a novel typing-derivation-directed CPS transformation that transforms an ATM-typable program with effect handlers to a simply-typable program without effect handlers.  Our CPS transformation is based on the one proposed by Kawamata et al.~\cite{DBLP:journals/pacmpl/KawamataUST24} but using typing-derivation dependency to eliminate parametric polymorphism.  The elimination was crucial for reducing the problem to the decidable problem of reachability for (non-polymorphic) finitary PCF.  As mentioned before, \cite{DBLP:journals/pacmpl/KawamataUST24} used higher-rank parametric polymorphism in an essential way to allow pure computations to be used in contexts where effectful ones are expected, and we have observed that this is precisely what subtyping is used for by the ATM type system.  We have adopted the ideas from the (sub)typing-derivation-directed CPS transformation for delimited control presented in the paper by Materzok and Biernacki~\cite{DBLP:conf/icfp/MaterzokB11} to design a new (sub)typing-derivation-directed CPS transformation for effect handlers.  The CPS transformation of \cite{DBLP:conf/icfp/MaterzokB11} does not consider effect handlers, and, as mentioned above, the CPS transformation of \cite{DBLP:journals/pacmpl/KawamataUST24} requires parametric polymorphism.  Our new CPS transformation makes a novel combination of the ideas from these prior CPS transformations.

A recent paper by Sekiyama and Unno~\cite{DBLP:journals/pacmpl/Sekiyama024} proves a similar but strictly weaker decidability result which essentially says that the typability in an ATM type system lacking subtyping is sufficient for decidability of the reachability for finitary PCF with effect handlers.  Our paper proves a stronger result that says that typability in the ATM type system that supports subtyping is sufficient for the decidability, and their result follows as an easy corollary of our result.  We note that subtyping makes a significant difference for an ATM type system.  Indeed, without subtyping, an ATM type system cannot, for example, type a program that uses a pure function in contexts expecting effectful computations with different answer-type modifications.  The lack of subtyping allowed their paper to use a more straightforward CPS transformation obtained by simply dropping parametric polymorphism from the CPS transformation of \cite{DBLP:journals/pacmpl/KawamataUST24}.  By contrast, our CPS transformation makes a novel use of the ideas from the (sub)typing-derivation-driven CPS transformation of \cite{DBLP:conf/icfp/MaterzokB11} to eliminate parametric polymorphism without losing the support for subtyping.
Additionally, we have disproved an erroneous claim made in \cite{DBLP:journals/pacmpl/Sekiyama024} regarding active effect handlers.  Namely, we have shown that the notion does not capture the decidability of reachability for programs with effect handlers, and that boundedness of their number is not the reason why ATM ensures the decidability of reachability.


\section{Conclusion}
\label{sec:conc}
We have studied the reachability problem for finitary PCF extended with effect handlers.  Recent work~\cite{DBLP:journals/pacmpl/LagoG24,DBLP:journals/pacmpl/Kobayashi25} have shown that the problem is undecidable, in stark contrast to the case without effect handlers for which the problem is known to be decidable~\cite{DBLP:conf/lics/Ong06,DBLP:conf/popl/Kobayashi09,DBLP:conf/fossacs/Tsukada014}.  In this paper, we have shown that the undecidability comes from the way the standard type systems for effect handlers are designed.  Concretely, we have shown that, perhaps surprisingly, extending the type system to allow ATM recovers decidability.  A corollary of our decidability result is that, perhaps counter-intuitively, there are program that are typable without ATM but untypable with it, and we have shown concrete examples of such programs.  Our decidability result was proven by a novel typing-derivation-driven CPS transformation, and as another application of the CPS transformation, we have shown that every ATM-typable recursive-definition-free finitary PCF programs terminates while the same does not hold for the simply-typable ones.  Finally, we have disproved a claim made in a recent paper~\cite{DBLP:journals/pacmpl/Sekiyama024} by showing that active effect handlers do not characterize the decidability of reachability for finitary PCF with effect handlers.
We foresee our decidability result to lay a foundation for developing verification methods for programs with effect handlers, just as the decidability result for reachability of finitary PCF has done such for programs without effect handlers.

\textit{Acknowledgments.}  We thank the anonymous reviewers for their suggestions.  We thank Yukiyoshi Kameyama, Naoki Kobayashi, Taro Sekiyama, and Hiroshi Unno for discussions on the work.  This work was supported by JSPS KAKENHI Grant Numbers JP23K24826 and JP20K20625.


\bibliographystyle{splncs04}
\bibliography{main}

\full{
\appendix
\section{$\vdashatm$ Typing Rules}
\label{app:atmtypingrulesfull}

Figure~\ref{fig:atmtypingrulesfull} shows the complete set of $\vdashatm$ typing rules.

\begin{figure}[h]
\[
\begin{array}{c}
  \infer[(\textsc{T-Unit})]{\Gamma \vdashatm () : \unitty}{} \ \
  \infer[(\textsc{T-Bool})]{\Gamma \vdashatm v : \boolty}{v \in \{\tttrue,\ttfalse\}} \\\\
  \infer[(\textsc{T-Var})]{\Gamma \vdashatm x : \tau}{x:\tau \in \Gamma} \ \
  \infer[(\textsc{T-Lam})]{\Gamma \vdashatm \lambda x.c : \tau \rightarrow \rho}{\Gamma,x:\tau \vdashatm c : \rho} \ \
  \infer[(\textsc{T-Rec})]{\Gamma \vdashatm \ttrec{x}{v} : \tau}{\Gamma, x:\tau \vdashatm v : \tau} \\\\
  \infer[(\textsc{T-If})]{\Gamma \vdashatm \ttif{v}{c_1}{c_2} : \rho}{\Gamma \vdashatm v : \boolty & \Gamma \vdashatm c_1 : \rho & \Gamma \vdashatm c_2 : \rho} \\\\
  \infer[(\textsc{T-App})]{\Gamma \vdashatm v_1\;v_2 : \rho}{\Gamma \vdashatm v_1 : \tau\rightarrow\rho & \Gamma \vdashatm v_2 : \tau} \\\\
  
\infer[(\textsc{T-LetP})]
      {\Gamma \vdashatm \ttlet{x}{c_1}{c_2} : \tau_2/\pureeff}
      {\Gamma \vdashatm c_1 : \tau_1/\pureeff & \Gamma,x:\tau_1 \vdashatm c_2 : \tau_2/\pureeff}\\\\
\infer[(\textsc{T-LetIp})]
      {\Gamma \vdashatm \ttlet{x}{c_1}{c_2} : \tau_2/\rho_2\Rightarrow\rho_1'}
      {\Gamma \vdashatm c_1 : \tau_1/\rho_1\Rightarrow\rho_1' & \Gamma,x:\tau_1 \vdashatm c_2 : \tau_2/\rho_2\Rightarrow\rho_1}\\\\
\infer[(\textsc{T-Ret})]
      {\Gamma \vdashatm \ttreturn\;v : \tau/\pureeff}
      {\Gamma \vdashatm v : \tau} \ \
\infer[(\textsc{T-Op})]
      {\Gamma \vdashatm \itop\;v : \tau'/\rho_1\Rightarrow\rho_2}
      {\Sigma(\itop) = \tau\rightarrow\tau'/\rho_1\Rightarrow\rho_2 & \Gamma \vdashatm v : \tau}\\\\
\infer[(\textsc{T-Hdlr)}]
      {\Gamma \vdashatm \{ \ttreturn\;x = c, \overline{\itop_i\;x_i\;k_i = c_i}\}}
      {\overline{\Sigma(\itop_i) = \tau_i\rightarrow\tau_i'/\rho_i\Rightarrow\rho_i'} & \overline{\Gamma, x_i:\tau_i, k_i:\tau_i'\rightarrow\rho_i \vdashatm c_i : \rho_i'}}\\\\
\infer[(\textsc{T-Han})]
      {\Gamma \vdashatm \ttwithhandle{h}{c} : \rho'}
      {\Gamma \vdashatm h & \Gamma \vdashatm c : \tau/\rho\Rightarrow\rho' & \Gamma,x:\tau \vdashatm c' : \rho & \ttreturn\;x = c' \in h}\\\\
\infer[(\textsc{T-VSub})]
      {\Gamma \vdashatm v : \tau'}
      {\Gamma \vdashatm v : \tau & \tau \leq \tau'}\ \
\infer[(\textsc{T-CSub})]
      {\Gamma \vdashatm c : \rho'}
      {\Gamma \vdashatm c : \rho & \rho \leq \rho'} 
\end{array}
\]
\caption{The typing rules of the ATM type system $\vdashatm$.}
\label{fig:atmtypingrulesfull}
\end{figure}

\section{$\vdashatm$ Subtyping Rules}
\label{app:atmsubtypingrulesfull}

Figure~\ref{fig:atmsubtypingrulesfull} shows the complete set of $\vdashatm$ subtyping rules.

\begin{figure}[h]
\[
\begin{array}{c}
\infer[(\textsc{S-Base})]
      {b \leq b}
      {}
      \ \
\infer[(\textsc{S-Arr})]
      {\tau_1 \rightarrow \rho_1 \leq \tau_2 \rightarrow \rho_2}
      {\tau_2 \leq \tau_1 & \rho_1 \leq \rho_2}
      \ \
\infer[(\textsc{S-Pure})]
      {\tau_1/\pureeff \leq \tau_2/\pureeff}
      {\tau_1 \leq \tau_2}
      \\\\
\infer[(\textsc{S-Ipure})]
      {\tau_1/\rho_1\Rightarrow\rho_1' \leq \tau_2/\rho_2\Rightarrow\rho_2'}
      {\tau_1 \leq \tau_2 & \rho_2 \leq \rho_1 & \rho_1' \leq \rho_2'}
      \ \ 
\infer[(\textsc{S-Embed})]
      {\tau_1/\pureeff \leq \tau_2/\rho_1\Rightarrow\rho_2}
      {\tau_1 \leq \tau_2 & \rho_1 \leq \rho_2}
\end{array}
\]
\caption{The subtyping rules of $\vdashatm$.}
\label{fig:atmsubtypingrulesfull}
\end{figure}

\section{Extending $\vdashst$-Typable $\plnamenoeff$ with Records}
\label{app:recordext}
We extend the operational semantics by adding the following rule.
\[
\infer[(\textsc{E-Proj})]
{\{ \overline{l_i = v_i}\}.l_i \rightarrow v_i }{}
\]
We extend the simple type system $\vdashst$ by adding the following rules.
\[
\infer[(\textsc{St-Record})]
      {\Gamma \vdashst \{ \overline{l_i = v_i} \} : \{ \overline{l_i:\sigma_i} \}}{\overline{\Gamma \vdashst v_i : \sigma_i} } \ \
\infer[(\textsc{St-Proj})]
      {\Gamma \vdashst v.l_i : \sigma_i}
      {\Gamma \vdashst v : \{\overline{l_i : \sigma_i} \}}
\]

\section{CPS Transformation for Subtyping Derivations}
\label{app:cpssubtyping}

Figure~\ref{fig:cpssubtypingfull} shows the complete set of CPS transformation rules for subtyping derivations.

\begin{figure}[h]
\[
\begin{array}{rcl}
  \sembrack{b \leq b} & = & \lambda x.x \\
  \sembrack{\tau_1\rightarrow\rho_1 < \tau_2\rightarrow\rho_2} & = & \lambda f.\lambda x.\sembrack{\rho_1 \leq \rho_2}@(f@(\sembrack{\tau_2 \leq \tau_1}@x))\\
  \sembrack{\tau_1/\pureeff \leq \tau_2/\pureeff} & = & \lambda x.\sembrack{\tau_1 \leq \tau_2}@x \\
  \sembrack{\tau_1/\rho_1\Rightarrow\rho_1' \leq \tau_2/\rho_2\Rightarrow\rho_2'} & & \\
  \multicolumn{3}{l}{\hspace{2cm} = \lambda x.\lambda h.\lambda k.\sembrack{\rho_1' \leq \rho_2'}@(x@h@\lambda y.\sembrack{\rho_2 \leq \rho_1}@(k@(\sembrack{\tau_1 \leq \tau_2}@y)))}\\
  \sembrack{\tau_1/\pureeff \leq \tau_2/\rho_1\Rightarrow\rho_2} & = & \lambda x.\lambda h.\lambda k.\sembrack{\rho_1 \leq \rho_2}@(k@(\sembrack{\tau_1 \leq \tau_2}@x)
\end{array}
\]
\caption{The CPS transformation rules for subtyping derivations.}
\label{fig:cpssubtypingfull}
\end{figure}

\section{CPS Transformation for Typing Derivations}
\label{app:cpstyping}

Figure~\ref{fig:cpstypingfull} shows the complete set of CPS transformation rules for typing derivations.

\begin{figure}[h]
\[
\begin{array}{l}
  \sembrackstretch{\Gamma \vdashatm () : \unitty} = () \ \ \ \
  \sembrackstretch{\dfrac{v \in \{\tttrue,\ttfalse\}}{\Gamma \vdashatm v : \boolty}} = v \ \ \ \
  \sembrackstretch{\dfrac{x:\tau \in \Gamma}{\Gamma \vdashatm x : \tau}} = x\\\\
  \sembrackstretch{\dfrac{\Gamma,x:\tau \vdashatm c : \rho}{\Gamma \vdashatm \lambda x.c : \tau \rightarrow \rho}} = \lambda x.\sembrackstretch{\Gamma,x:\tau \vdashatm c : \rho} \\\\
  \sembrackstretch{\dfrac{\Gamma, x:\tau \vdashatm v : \tau}{\Gamma \vdashatm \ttrec{x}{v} : \tau}} = \ttrec{x}{\sembrackstretch{\Gamma, x:\tau \vdashatm v : \tau}}\\\\
  \sembrackstretch{\dfrac{\Gamma \vdashatm v : \boolty \ \ \Gamma \vdashatm c_1 : \rho \ \ \Gamma \vdashatm c_2 : \rho}{\Gamma \vdashatm \ttif{v}{c_1}{c_2}}} \\
  \hspace{2cm} = \ttif{\sembrackstretch{\Gamma \vdashatm v : \boolty}}{\sembrackstretch{\Gamma \vdashatm c_1 : \rho}}{\sembrackstretch{\Gamma \vdashatm c_2 : \rho}}\\\\
  \sembrackstretch{\dfrac{\Gamma \vdashatm v_1 : \tau\rightarrow\rho \ \ \Gamma \vdashatm v_2 : \tau}{\Gamma \vdashatm v_1\;v_2 : \rho}} = \sembrackstretch{\Gamma \vdashatm v_1 : \tau\rightarrow\rho}\;\sembrackstretch{\Gamma \vdashatm v_2 : \tau}\\\\
  \sembrackstretch{\dfrac{\Gamma \vdashatm c_1 : \tau_1/\pureeff \ \ \Gamma,x:\tau_1 \vdashatm c_2 : \tau_2/\pureeff}{\Gamma \vdashatm \ttlet{x}{c_1}{c_2} : \tau_2/\pureeff}} \\
  \hspace{3.5cm} = \left(\lambda x.\sembrackstretch{\Gamma,x:\tau_1 \vdashatm c_2 : \tau_2/\pureeff} \right)\;\sembrackstretch{\Gamma \vdashatm c_1 : \tau_1/\pureeff}\\\\
  \sembrackstretch{\dfrac{\Gamma \vdashatm c_1 : \tau_1/\rho_1\Rightarrow\rho_1' \ \ \Gamma,x:\tau_1 \vdashatm c_2 : \tau_2/\rho_2\Rightarrow\rho_1}{\Gamma \vdashatm \ttlet{x}{c_1}{c_2} : \tau_2/\rho_2\Rightarrow\rho_1'}}\\
  = \lambda h.\lambda k.\sembrackstretch{\Gamma \vdashatm c_1 : \tau_1/\rho_1\Rightarrow\rho_1'}@h@\left(\lambda x.\sembrackstretch{\Gamma,x:\tau_1 \vdashatm c_2 : \tau_2/\rho_2\Rightarrow\rho_1}@h@k\right)\\\\
  \sembrackstretch{\dfrac{\Gamma \vdashatm v : \tau}{\Gamma \vdashatm \ttreturn\;v : \tau/\pureeff}} = \sembrackstretch{\Gamma \vdashatm v : \tau}\\\\
  \sembrackstretch{\dfrac{\Sigma(\itop) = \tau\rightarrow\tau'/\rho_1\Rightarrow\rho_2 \ \ \Gamma \vdashatm v : \tau}{\Gamma \vdashatm \itop\;v : \tau'/\rho_1\Rightarrow\rho_2}} = \lambda h.\lambda k. h.\itop\;\sembrack{\Gamma \vdashatm v : \tau}\;k\\\\
    \left\llbracket \dfrac{\overline{\Sigma(\itop_i) = \tau_i\rightarrow\tau_i'/\rho_i\Rightarrow\rho_i'} \ \ \overline{\Gamma, x_i:\tau_i, k_i:\tau_i'\rightarrow\rho_i \vdashatm c_i : \rho_i'}}{\Gamma \vdashatm \{ \ttreturn\;x = c, \overline{\itop_i\;x_i\;k_i = c_i}\}} \right\rrbracket \\
    \hspace{3.5cm} = \left\{ \overline{\itop_i = \lambda x_i.\lambda k_i.\sembrack{\Gamma, x_i:\tau_i, k_i : \tau_i'\rightarrow\rho_i \vdashatm c_i : \rho_i'}} \right\}\\\\
      \sembrackstretch{\dfrac{\Gamma \vdashatm h \ \ \Gamma \vdashatm c : \tau/\rho\Rightarrow\rho' \ \ \Gamma,x:\tau \vdashatm c' : \rho \ \ \ttreturn\;x = c' \in h}{\Gamma \vdashatm \ttwithhandle{h}{c} : \rho'}}\\
      \hspace{2.5cm} = \sembrackstretch{\Gamma \vdashatm c : \tau/\rho\Rightarrow\rho'}@\sembrackstretch{\Gamma \vdashatm h}@\lambda x.\sembrackstretch{\Gamma,x:\tau \vdashatm c' : \rho}\\\\
      \sembrackstretch{\dfrac{\Gamma \vdashatm c : \rho \ \ \rho \leq \rho'}{\Gamma \vdashatm c : \rho'}} = \sembrackstretch{\rho \leq \rho'}@\sembrackstretch{\Gamma \vdashatm c : \rho} \\\\
      \sembrackstretch{\dfrac{\Gamma \vdashatm v : \tau \ \ \tau \leq \tau'}{\Gamma \vdashatm v : \tau'}} = \sembrackstretch{\tau \leq \tau'}@\sembrackstretch{\Gamma \vdashatm v : \tau}
\end{array}
\]
\caption{The CPS transformation rules for typing derivations.}
\label{fig:cpstypingfull}
\end{figure}

\section{Proof of Theorem~\ref{thm:simulation}}
\label{app:simulation}

We first show the basic type soundness property for $\vdashatm$ by proving the standard preservation and progress properties~\cite{DBLP:journals/iandc/WrightF94}.  Besides sanity checking that $\vdashatm$ enjoys the usual correctness properties expected for a type system, this is useful for proving Theorem~\ref{thm:simulation} because our CPS transformation is typing-derivation-driven and the fact that typability is preserved allows us to pick a typing derivation for $c'$ to CPS transform it when $\vdashatm c : \rho$ and $c \rightarrow c'$.

\begin{lemma}[Progress]
If $\vdashatm c : \rho$, then either
\begin{itemize}
  \item $c \rightarrow c'$ for some $c'$
  \item $c = \ttreturn\;v$ for some $v$, or
  \item $c = E[\itop\;v]$ for some $E, v, \itop$.
\end{itemize}
\end{lemma}
\begin{proof}
  Immediate from the definitions of $\vdashatm$ and $\rightarrow$.
  \qed
\end{proof}
The following is also immediate.
\begin{lemma}[Substitution]
\label{lem:subst}
  If $\Gamma,x:\tau \vdashatm c : \rho$ and $\Gamma \vdashatm v : \tau$ then
  $\Gamma \vdashatm c[v/x] : \rho$.
\end{lemma}
We now show the preservation property.
\begin{lemma}[Preservation]
\label{lem:preservation}
If $\Gamma \vdashatm c : \rho$ and $c \rightarrow c'$, then $\Gamma \vdashatm c' : \rho$.
\end{lemma}
\begin{proof}
  We prove by induction on the derivation of $c \rightarrow c'$.  The cases besides (\textsc{E-Op}) are immediate from Lemma~\ref{lem:subst}.  For the case (\textsc{E-Op}), we may assume without loss of generality that the derivation $\Gamma \vdashatm c : \rho$ is of the form
\[
\infer[(\textsc{T-Han})]{\Gamma \vdashatm \ttwithhandle{h}{E[\itop\;v]} : \rho}
  {\Gamma \vdashatm h & \Gamma,x:\tau \vdashatm c'' : \rho' & \Gamma \vdashatm E[\itop\;v] : \tau/\rho'\Rightarrow \rho}
\]
where $\ttreturn\;x = c'' \in h$ and the subderivation $\Gamma \vdashatm E[\itop\;v] : \tau/\rho'\Rightarrow \rho$ is of the form
\[
\infer{\Gamma \vdashatm E[\itop\;v] : \tau/\rho'\Rightarrow \rho}
      {\begin{array}{c}\mathcal{D}\\\vdots\end{array}}
\]
containing a subderivation $\mathcal{D}$ of the form
\[
\infer[(\textsc{T-Op})]
      {\Gamma \vdashatm \itop\;v : \tau'/\rho_2\Rightarrow\rho}
      {\Gamma(\itop) = \tau_2\rightarrow\tau'/\rho_2\Rightarrow\rho & \Gamma \vdashatm v : \tau_2}
      \]
By replacing this $\mathcal{D}$, by the following derivation $\mathcal{D'}$
\[
\infer[(\textsc{T-Csub})]
      {\Gamma, y:\tau_2 \vdashatm \ttreturn\;y : \tau_2/\rho_2\Rightarrow\rho_2}
      {\Gamma, y:\tau_2 \vdashatm \ttreturn\;y : \tau_2/\pureeff & \tau_2/\pureeff \leq \tau_2/\rho_2\Rightarrow\rho_2}
\]
we obtain the derivation
\[
\infer[(\textsc{T-Han})]{\Gamma' \vdashatm \ttwithhandle{h}{E[\ttreturn\;y] : \rho_2}}
  {\Gamma' \vdashatm h & \Gamma',x:\tau \vdashatm c'' : \rho' & \Gamma' \vdashatm E[\ttreturn\;y] : \tau/\rho'\Rightarrow \rho_2}
  \]
  where $\Gamma' = \Gamma,y:\tau_2 $.
Therefore, by (\textsc{T-Hdlr}) and Lemma~\ref{lem:subst}, we have
\[
\Gamma \vdashatm c[v/x,\lambda y.\ttwithhandle{h}{E[\ttreturn\:y]}/k] : \rho
\]
\qed
\end{proof}
We obtain the type soundness property as a corollary of the progress and the preservation properties.
\begin{corollary}[Type Soundness]
\label{cor:soundness}
If $c$ is $\vdashatm$-typable and $c \rightarrow^* c'$, then either
\begin{itemize}
  \item $c' \rightarrow c''$ for some $\vdashatm$-typable $c''$
  \item $c' = \ttreturn\;v$ for some $v$, or
  \item $c' = E[\itop\;v]$ for some $E, v, \itop$.
\end{itemize}
\end{corollary}

The following lemma states that substitution and CPS transformation commute.  
\begin{lemma}
\label{lem:commute}
Let $\mathcal{D}_c$ be a typing derivation of $\Gamma, x:\tau \vdashatm c: \rho$, and $\mathcal{D}_v$ be a typing derivation of $\Gamma \vdashatm v : \tau$.  Then, $\sembrack{\mathcal{D}_c[\mathcal{D}_v/x]} = \sembrack{\mathcal{D}_c}[\sembrack{\mathcal{D}_v}/x]$ where $\mathcal{D}_c[\mathcal{D}_v/x]$ denotes the derivation obtained by replacing occurrences of $\Gamma, x:\tau \vdashatm x: \tau$ by $\mathcal{D}_v$.
\end{lemma}
\begin{proof}
  By induction on the derivation $\mathcal{D}_c$
  \qed
\end{proof}

Next we prove that one step of the source program evaluation can be simulated by the transformed program.
\begin{lemma}[One-Step Forward Simulation]
\label{lem:onestepforward}
Let $\mathcal{D}$ be a derivation of  $\Gamma \vdashatm c : \rho$ and $c \rightarrow c'$.  If $\rho$ is pure, then there is a derivation $\mathcal{D'}$ of $\Gamma \vdashatm c' : \rho$ such that $\sembrack{\mathcal{D}}\rightarrow^+\sembrack{\mathcal{D}'}$.  Otherwise, $\rho$ is effectful, and there is a derivation $\mathcal{D'}$ of $\Gamma \vdashatm c' : \rho$ such that $\sembrack{\mathcal{D}}@v_h@v_k \rightarrow^+ \sembrack{\mathcal{D}'}@v_h@v_k$ for any values $v_h$ and $v_k$.
\end{lemma}
\begin{proof}
Let us fix a derivation $\mathcal{D}'$ for $c'$ constructed by Lemma~\ref{lem:preservation}.  We prove by simultaneous induction on the derivation of $c \rightarrow c'$ and the derivation $\mathcal{D}$.

First, we consider the case the root of $\mathcal{D}$ is an instance of (\textsc{T-CSub}).  $\mathcal{D}$ must be of the form
\[
\infer[(\textsc{T-CSub})]
      {\Gamma \vdashatm c : \rho}
      {\Gamma \vdashatm c : \rho' & \rho' \leq \rho}
\]
Therefore, $\sembrack{\mathcal{D}} = \sembrack{\rho' \leq \rho} @ \sembrack{\Gamma \vdashatm c : \rho'}$ and $\sembrack{\mathcal{D'}} = \sembrack{\rho' \leq \rho} @ \sembrack{\Gamma \vdashatm c' : \rho'}$.
If $\rho$ is pure, then we have
\[
  \sembrack{\mathcal{D}} = \sembrack{\rho' \leq \rho} @ \sembrack{\Gamma \vdashatm c : \rho'} \rightarrow^+ \sembrack{\rho' \leq \rho} @ \sembrack{\Gamma \vdashatm c' : \rho'} = \sembrack{\mathcal{D'}}
\]
Otherwise, $\rho$ is effectful and is of the form $\tau/\rho_1\Rightarrow\rho_2$ for some $\tau$, $\rho_1$, and $\rho_2$.  If $\rho'$ is a pure computation type of the form $\tau'/\pureeff$, then
\[
\begin{array}{rll}
  \sembrack{\mathcal{D}}@v_h@v_k & = & \sembrack{\rho' \leq \rho}@\sembrack{\Gamma \vdashatm c : \rho'}@v_h@v_k \\
  & = & \sembrack{\rho_1 \leq \rho_2}@(v_k\;(\sembrack{\tau' \leq \tau}@\sembrack{\Gamma \vdashatm c : \rho'})) \\
  & \rightarrow^+ & \sembrack{\rho_1 \leq \rho_2}@(v_k\;(\sembrack{\tau' \leq \tau}@\sembrack{\Gamma \vdashatm c' : \rho'})) \\
  & = & \sembrack{\rho' \leq \rho}@\sembrack{\Gamma \vdashatm c' : \rho'}@v_h@v_k\\
  & = & \sembrack{\mathcal{D'}}@v_h@v_k
\end{array}
\]
Otherwise, $\rho'$ is also effectful and is of the form $\tau'/\rho_1'\Rightarrow\rho_2'$.  Then,
\[
\begin{array}{rll}
  \sembrack{\mathcal{D}}@v_h@v_k & = & \sembrack{\rho' \leq \rho}@\sembrack{\Gamma \vdashatm c : \rho'}@v_h@v_k \\
  & = & \sembrack{\rho_2' \leq \rho_2}@(\sembrack{\Gamma \vdashatm c : \rho'}@v_h@\lambda y.\sembrack{\rho_1 \leq \rho_1'}@(k@(\sembrack{\tau\leq\tau'}@y))\\
  & \rightarrow^+ & \sembrack{\rho_2' \leq \rho_2}@(\sembrack{\Gamma \vdashatm c' : \rho'}@v_h@\lambda y.\sembrack{\rho_1 \leq \rho_1'}@(k@(\sembrack{\tau\leq\tau'}@y))\\
  & = & \sembrack{\rho' \leq \rho}@\sembrack{\Gamma \vdashatm c' : \rho'}@v_h@v_k\\
  & = & \sembrack{\mathcal{D'}}@v_h@v_k
\end{array}
\]
In what follows, we assume that the last rule used in $\mathcal{D}$ is not (\textsc{T-CSub}).  We prove by case analysis on the last rule of $c \rightarrow c'$.

Suppose that it is (\textsc{E-Let}).  Then, $c \rightarrow c'$ is of the form
\[
\infer[(\textsc{E-Let})]
      {\ttlet{x}{c_1}{c_2} \rightarrow \ttlet{x}{c_1'}{c_2}}
      {c_1 \rightarrow c_1'}
\]
If $\rho$ is pure and is of the form $\tau/\pureeff$, then $\mathcal{D}$ must be of the form
\[
\infer[(\textsc{T-LetP})]
      {\Gamma \vdashatm \ttlet{x}{c_1}{c_2} : \tau/\pureeff}
      {\Gamma \vdashatm c_1 : \tau_1/\pureeff & \Gamma,x:\tau_1 \vdashatm c_2 : \tau/\pureeff }
\]
Therefore,
\[
\begin{array}{rll}
  \sembrack{\mathcal{D}} & = & (\lambda x.\sembrack{\Gamma,x:\tau_1 \vdashatm c_2 : \tau/\pureeff})\;\sembrack{\Gamma \vdashatm c_1 : \tau_1/\pureeff} \\
  & \rightarrow^+ & (\lambda x.\sembrack{\Gamma,x:\tau_1 \vdashatm c_2 : \tau/\pureeff})\;\sembrack{\Gamma \vdashatm c_1' : \tau_1/\pureeff}\\
  & = & \sembrack{\mathcal{D'}}
\end{array}
\]
If $\rho$ is effectful and is of the form $\tau/\rho_1\Rightarrow\rho_2$, then $\mathcal{D}$ must be of the form
\[
\infer[(\textsc{T-LetIp})]
      {\Gamma \vdashatm \ttlet{x}{c_1}{c_2} : \tau/\rho_1\Rightarrow\rho_2}
      {\Gamma \vdashatm c_1 : \tau_1/\rho'\Rightarrow\rho_2 & \Gamma,x:\tau_1 \vdashatm c_2 : \tau/\rho_1\Rightarrow\rho'}
\]
Therefore,
\[
\begin{array}{ll}
  \multicolumn{2}{l}{\sembrack{\mathcal{D}}@v_h@v_k}\\
  \hspace{0.3cm}= & \sembrack{\Gamma \vdashatm c_1 : \tau_1/\rho'\Rightarrow\rho_2}@v_h@(\lambda x.\sembrack{\Gamma,x:\tau_1 \vdashatm c_2 : \tau/\rho_1\Rightarrow\rho'}@v_h@v_k)\\
  \hspace{0.3cm}\rightarrow^+ & \sembrack{\Gamma \vdashatm c_1' : \tau_1/\rho'\Rightarrow\rho_2}@v_h@(\lambda x.\sembrack{\Gamma,x:\tau_1 \vdashatm c_2 : \tau/\rho_1\Rightarrow\rho'}@v_h@v_k)\\
  \hspace{0.3cm}= & \sembrack{\mathcal{D'}}@v_h@v_k
\end{array}
\]

Next, suppose that the last rule is (\textsc{E-Ret}).  Then, $c \rightarrow c'$ must be of the form
\[
\infer{\ttlet{x}{\ttreturn\;v}{c_2} \rightarrow c_2[v/x]}
      {}
\]
If $\rho$ is a pure type of the form $\tau/\pureeff$ then $\mathcal{D}$ must be of the form
\[
\infer{\Gamma \vdashatm \ttlet{x}{\ttreturn\;v}{c_2} : \tau/\pureeff}
      {\infer{\Gamma \vdashatm \ttreturn\;v: \tau_1/\pureeff}{\Gamma \vdashatm v : \tau_1} & \Gamma,x:\tau_1 \vdashatm c_2 : \tau/\pureeff}
\]
Therefore,
\[
\begin{array}{rll}
\sembrack{\mathcal{D}} & = & (\lambda x.\sembrack{\Gamma,x:\tau_1 \vdashatm c_2 : \tau/\pureeff})\;\sembrack{\Gamma \vdashatm v : \tau_1} \\
& \rightarrow & \sembrack{\Gamma,x:\tau_1 \vdashatm c_2 : \tau/\pureeff}[\sembrack{\Gamma \vdashatm v : \tau_1}/x] \\
& = & \sembrack{\mathcal{D'}}
\end{array}
\]
where the last equation follows by Lemma~\ref{lem:commute}.
If $\rho$ is an effectful type of the form $\tau/\rho_1\Rightarrow\rho_2$ then $\mathcal{D}$ must be of the form
\[
\infer{\Gamma \vdashatm \ttlet{x}{\ttreturn\;v}{c_2} : \tau/\rho_1\Rightarrow \rho_2}
      {\infer{\Gamma \vdashatm \ttreturn\;v: \tau_1/\rho'\Rightarrow\rho_2}{
          \mathcal{D}_1 & \mathcal{D}_2}
        &
        \Gamma,x:\tau_1 \vdashatm c_2 : \tau/\rho_1\Rightarrow\rho'
      }
\]
where $\mathcal{D}_1$ and $\mathcal{D}_2$ are as follows
\[
\setlength{\arraycolsep}{0.5cm}
\begin{array}{cc}
\infer{\Gamma \vdashatm \ttreturn\;v : \tau_1/\pureeff}
      {\Gamma \vdashatm v : \tau_1} 
&
\infer{\tau_1/\pureeff \leq \tau_1/\rho'\Rightarrow\rho_2}
      {\rho' \leq \rho_2}
\end{array}
\]
Therefore,
\[
\begin{array}{ll}
  \multicolumn{2}{l}{\sembrack{\mathcal{D}}@v_h@v_k} \\
  \hspace{0.5cm} = & \sembrack{\rho' \leq \rho_2}@((\lambda x.\sembrack{\Gamma,x:\tau_1 \vdashatm c_2 : \tau/\rho_1\Rightarrow\rho'}@v_h@v_k)\;\sembrack{\Gamma \vdashatm v : \tau_1})\\
 \hspace{0.5cm} \rightarrow & \sembrack{\rho' \leq \rho_2}@(\sembrack{\Gamma,x:\tau_1 \vdashatm c_2 : \tau/\rho_1\Rightarrow\rho'}[\sembrack{\Gamma \vdashatm v : \tau_1}/x]@v_h@v_k) \\
 \hspace{0.5cm} = & \sembrack{\mathcal{D}'}@v_h@v_k
\end{array}
\]
where the last equation follows by Lemma~\ref{lem:commute} again.
The other cases can be similarly shown by applying Lemma~\ref{lem:commute}.
\qed      
\end{proof}
The forward direction of the simulation theorem, that is, item \ref{item:forward} of Theorem~\ref{thm:simulation}, follows by repeated applications of Lemma~\ref{lem:onestepforward}.

Next we prove the backward direction, that is, item \ref{item:backward} of Theorem~\ref{thm:simulation}.
Because $c$ is $\vdashatm$-typable, by Corollary~\ref{cor:soundness}, the evaluation of $c$ either (1) diverges, (2) gets stuck by reaching some $E[\itop\;v]$, or (3) halts by returning a value.  Note that the evaluation relation $\rightarrow$ is deterministic.
If (1) is true, then by the forward direction of the theorem that we have just shown, it must be the case that the evaluation of $\sembrack{\;\vdashatm c : \tau/\pureeff}$ also diverges.
Next, suppose that (2) is true.  Then, by Lemma~\ref{lem:preservation}, it must be the case that $\vdashatm E[\itop\;v] : \tau/\pureeff$, but such a typing is not possible because $E[\itop\;v]$ can only be given an effectful computation type.
Finally, suppose that (3) is true.  Let $c \rightarrow^* \ttreturn\;v$.  It suffices to show that $\sembrack{\;\vdashatm v : \tau} = v'$ where $\sembrack{\;\vdashatm c : \tau/\pureeff} \rightarrow^* \ttreturn\;v'$.  But, $\sembrack{\;\vdashatm c : \tau/\pureeff}  \rightarrow^* \sembrack{\;\vdashatm \ttreturn\;v : \tau/\pureeff}$ by Lemma~\ref{lem:onestepforward}.  Therefore, by the determinism of the evaluation relation, it must be the case that $\sembrack{\;\vdashatm v : \tau} = v'$.
This complete the proof of Theorem~\ref{thm:simulation}.

\section{Proof of Theorem~\ref{thm:minskymachine}}
\label{app:minskymachine}

We note that the proof of the theorem, as well as the construction of the class of programs $\mathit{MM}$, is an adaptation of the proof of the undecidability of the reachability problem for finitary PCF extended with exceptions given in \cite{DBLP:journals/pacmpl/Kobayashi25}.  Namely, we adapt their proof to our setting by implementing the exceptions in their programs by effect handlers, following the usual approach of implementing exceptions by effect handlers~\cite{DBLP:journals/entcs/Pretnar15}.\footnote{The paper \cite{DBLP:journals/pacmpl/Kobayashi25} contains a proof that reachability for finitary PCF extended with effect handlers is undecidable (with a non-ATM type system), but they use a different class of programs.  In particular, the restriction on operation parameters mentioned in Section~\ref{subsec:errorinsekiyamaunno} prevents them from using $\mathit{MM}$.}

First, we recall 2-register Minsky machines~\cite{minsky1967computation}.  A \emph{2-register Minsky machine} is a tuple $M = (Q,\delta,q_0)$ where $Q$ is a finite set of states, $q_0 \in Q$ is the initial state, $\delta$ is the transition function which maps each state $q \in Q$ to the an instruction of one of the following three forms:
\begin{description}
\item[$\textsf{Inc}\;r_i;\textsf{goto}\;q'$] Increment register $r_i$ ($i \in \{0,1\}$) and go to state $q'$ ($q' \in Q$).
\item[$\textsf{If}\;r_i=0\;\textsf{then}\;\textsf{goto}\;q_1\;\textsf{else}\;\textsf{dec}\;r_i;\textsf{goto}\;q_2$] Check the value of register $r_i$ ($i \in \{0,1\}$).  If the value is $0$ then go to state $q_1$ ($q_1 \in Q$).  Otherwise, decrement $r_i$ and go to state $q_2$ ($q_2 \in Q$).
\item[$\textsf{Halt}$] Halt the machine.
\end{description}
A \emph{configuration} of the machine is a triple $(q,n_0,n_1)$ where $q \in Q$ and $n_0$ and $n_1$ are non-negative integers denoting the values of the two registers.  The computation of the machine starts from the initial configuration $(q_0,0,0)$ and terminates when the $\textsf{Halt}$ instruction is reached.  The \emph{halting problem} for 2-counter Minsky machines is to decide if the computation of a given 2-counter Minsky machine terminates.
\begin{theorem}[\cite{minsky1967computation}]
The halting problem for 2-counter Minsky machines is undecidable.
\end{theorem}

It is easy to see that each program in $\mathit{MM}^n$ implements an $n$+1-state 2-register Minsky machine by noticing that the recursive functions $f_0, \dots, f_n$ represent the $n$+1 states with $f_0$ in particular representing the initial state $q_0$, each $c_{\it inc}^{i,j}$ implements the increment instruction that increments the register $r_i$ and goes to the state represented by $f_j$, and each $c_{\it dec}^{i,j,m}$ implements the check-and-decrement instruction that checks the register $r_i$, and goes to the state represented by $f_j$ if it is $0$ and otherwise decrements $r_i$ and goes to the state represented by $f_m$.  The halt instruction is implemented by $()$.
Therefore, a program in $\mathit{MM}$ returns $\tttrue$ iff the corresponding 2-register Minsky machine halts.  Thus, the halting problem for 2-register Minsky machines is reduced to the reachability problem for $\mathit{MM}$.  This completes the proof of Theorem~\ref{thm:minskymachine}.

\texcomment{
\begin{proof}
Let $c$ be a $\plnamenoeff$ program that is $\vdashst$-typable.  Then, a
$\vdashatm$ typing derivation for $c$ can be obtained by using pure effects for every type of a computation term in the $\vdashst$-typing derivation of $c$.  That is, a $\vdashst$ type ... can be seen as the $\vdashatm$ type ...
Conversely, let $c$ be a $\plnamenoeff$ program that is $\vdashatm$-typable.
Then, there exist a $\vdashatm$-typing derivation for $c$ that only uses pure effects.  Then, such a derivation can be converted to a $\vdashst$-typing derivation by making the aforementioned analogy.
\qed
\end{proof}
}

}{}

\end{document}